\documentclass[5p]{elsarticle}
\usepackage[english]{babel}

\usepackage{amsmath,amssymb,amsthm,amsfonts}
\usepackage{exscale}

\usepackage[T1]{fontenc}
\usepackage{lmodern}

\usepackage{bbm}

\usepackage{graphicx}

\usepackage[alsoload=binary]{siunitx}

\usepackage{shuffle}

\usepackage{listings}
\usepackage{xcolor}

\usepackage{hyperref}

\usepackage{tabularx}
\usepackage{booktabs}

\definecolor{mapleinput}{rgb}{0.5,0.0,0.0}
\definecolor{maplemath}{rgb}{0.0,0.0,1.0}
\definecolor{maplewarning}{cmyk}{0.0,1.0,.0.0,0.0}

\newcommand{\Filename}[1]{{\upshape\ttfamily #1}}

\lstnewenvironment{MapleInput}{%
\lstset{
	basicstyle=\bfseries\upshape\ttfamily\color{mapleinput},
	numberstyle=\bfseries\upshape\ttfamily\color{mapleinput},
	breaklines=true,
	columns=flexible,
	breakautoindent=false,
	breakindent=0pt,
	xleftmargin=4ex,
	numbers=left,
	stepnumber=1,
	numbersep=1.3ex,
	aboveskip=\smallskipamount,
	belowskip=\smallskipamount,
}%
}{%
}

\lstnewenvironment{MapleError}{%
\lstset{
	basicstyle=\upshape\ttfamily\color{maplewarning},
	breaklines=true,
	columns=flexible,
	breakautoindent=false,
	breakindent=0pt,
	xleftmargin=2ex,
	numbers=none,
	aboveskip=0pt,
	belowskip=\smallskipamount,
}%
}{%
}

\lstnewenvironment{MapleOutput}{%
\lstset{
	basicstyle=\upshape\ttfamily\color{maplemath},
	breaklines=true,
	columns=flexible,
	breakautoindent=false,
	breakindent=0pt,
	xleftmargin=2ex,
	numbers=none,
	aboveskip=0pt,
	belowskip=\smallskipamount,
}%
}{%
}

\newenvironment{MapleMath}{%
\color{maplemath}\upshape\rmfamily%
\setlength{\abovedisplayskip}{0ex}%
\setlength{\abovedisplayshortskip}{\abovedisplayskip}%
\setlength{\belowdisplayskip}{\medskipamount}%
\setlength{\belowdisplayshortskip}{0ex}%
\csname gather*\endcsname}{\csname endgather*\endcsname%
{\hrule height 0pt}%
\ignorespacesafterend}

\newcommand{\Catalan}{\operatorname{Catalan}}

\numberwithin{equation}{section}

\newcommand{\K}{
\mathbbm{K}%
}
\newcommand{\Q}{
\mathbbm{Q}%
}
\newcommand{\N}{
\mathbbm{N}%
}
\newcommand{\Z}{
\mathbbm{Z}%
}

\newcommand{\C}{
\mathbbm{C}%
}
\newcommand{\Halfplane}{
\mathbbm{H}%
}

\newcommand{\tp}{
\otimes
}

\newlength{\wurelwidth}

\newcommand{\wurel}[2][=]{\mathrel{\mathop{#1}_{\!\scalebox{0.5}{\makebox[\the\wurelwidth]{#2}}\!}}}

\newcommand{\period}{\mathcal{P}}

\newcommand{\hide}[1]{}

\newcommand{\convolution}{\star}

\newcommand{\dd}[1][]{\mathrm{d}^{#1}}
\newcommand{\cupdot}{\mathbin{\dot{\cup}}}
\newcommand{\restrict}[2]{%
{\left. #1 \right|}_{#2}%
}

\DeclareMathOperator{\im}{im}

\newcommand{\defas}{
\mathrel{\mathop:}=
}

\newcommand{\set}[1]{
\left\{ #1 \right\}
}
\newcommand{\setexp}[2]{
\left\{ #1\!:\ #2 \right\}
}
\newcommand{\abs}[1]{
\left\lvert #1 \right\rvert
}

\newcommand{\id}{
\mathrm{id}
}

\newcommand{\Hyper}[1]{L_{#1}}
\newcommand{\letter}[1]{\omega_{#1}}
\DeclareMathOperator{\Li}{Li}
\newcommand{\mzv}[2][]{\zeta^{#1}_{#2 }}

\newcommand{\iInt}{\int}
\newcommand{\concat}{\star}

\DeclareMathOperator{\sdd}{sdd}
\newcommand{\loops}[1]{\abs{#1}}

\newcommand{\Graph}[2][1.0]{%
\vcenter{\hbox{\includegraphics[scale=#1]{Graphs/#2}}}%
}
\newcommand{\BBr}[2]{B_{#1,#2}}
\newcommand{\BBt}[2]{\hat{B}_{#1,#2}}
\newcommand{\Dim}{D}

\newcommand{\SP}{\alpha}
\newcommand{\EP}{a}
\newcommand{\EPE}{\nu}
\newcommand{\EPZ}{A}

\newcommand{\phipol}{\varphi}
\newcommand{\psipol}{\psi}

\newcommand{\forestpolynom}[1]{\Phi^{#1}}

\newcommand{\Maple}{\texttt{\textup{Maple}}}
\newcommand{\JaxoDraw}{\texttt{\textup{JaxoDraw}}}
\newcommand{\zetaprocedures}{\texttt{\textup{zeta\_procedures}}}

\newcommand{\HyperProg}{\texttt{\textup{HyperInt}}}

\newcommand{\AnaReg}[2]{\Reglim_{#1 \rightarrow #2}}
\newcommand{\WordReg}[2]{\WordReglim\nolimits_{#1}^{#2}}
\newcommand{\ReglimWord}[2]{\WordReglim_{#1 \rightarrow #2}}
\DeclareMathOperator*{\Reglim}{Reg}
\DeclareMathOperator*{\WordReglim}{reg}
\newcommand{\imag}{i}
\newcommand{\regulars}{\mathcal{O}}
\newcommand{\code}[1]{\mbox{\texttt{\textup{#1}}}}
\newcommand{\WordTransformation}[1]{\Phi_{#1}}
\newcommand{\emptyWord}{1}
\DeclareMathOperator{\leadingCoefficient}{lead}
\DeclareMathOperator{\Imaginaerteil}{Im}
\DeclareMathOperator{\Realteil}{Re}

\newcommand{\Hlog}[2]{\operatorname{Hlog}\left( #1, \left[ #2 \right] \right)}
\newcommand{\Mpl}[2]{\operatorname{Mpl}\left( \left[ #1 \right], \left[ #2 \right] \right)}

\newtheorem{theorem}{Theorem}[section]
\newtheorem{definition}[theorem]{Definition}
\newtheorem{lemma}[theorem]{Lemma}
\newtheorem{corollary}[theorem]{Corollary}

\newtheorem{example}[theorem]{Example}
\newtheorem{remark}[theorem]{Remark}

\biboptions{sort&compress}

\newcounter{bla}

\newcommand{\mytitle}{Algorithms for the symbolic integration of hyperlogarithms with applications to Feynman integrals}

\hypersetup{%
	pdftitle = {\mytitle},
	pdfauthor = {Erik Panzer}
}

\journal{Computer Physics Communications}

\begin{document}

\begin{frontmatter}



\title{\mytitle}

\author{Erik Panzer}
\ead{panzer@mathematik.hu-berlin.de}
\address{Humboldt-Universit\"{a}t zu Berlin, Institut f\"{u}r Physik, Newton Stra{\ss}e 15, 12489 Berlin, Germany}

\begin{abstract}%
	We provide algorithms for symbolic integration of hyperlogarithms multiplied by rational functions, which also include multiple polylogarithms when their arguments are rational functions.
	These algorithms are implemented in {\Maple} and we discuss various applications.
	In particular, many Feynman integrals can be computed by this method.
\end{abstract}


\end{frontmatter}

\pdfbookmark[1]{Program summary}{program-summary}
\section*{Program summary}

\begin{small}
\noindent
{\em Manuscript Title:}
	\mytitle
\\
{\em Authors:}
	Erik Panzer
\\
{\em Program Title:}
	\HyperProg
\\
{\em Journal Reference:}
\\
{\em Catalogue identifier:}
\\
{\em Licensing provisions:}
	GNU General Public License, version 3
\\
{\em Programming language:}
	{\Maple} [1], version 16 or higher
\\
{\em Computer:}
	Any that supports {\Maple}
\\
{\em Operating system:}
	Any that supports {\Maple}
\\
{\em RAM:}
	Highly problem dependent; from a few \si{\mebi\byte} to many \si{\gibi\byte} 
\\
{\em Number of processors used:}
\\
{\em Supplementary material:}
	Example worksheet \Filename{Manual.mw} explaining most features provided and including plenty of examples of Feynman integral computations.
\\
{\em Keywords:}
	hyperlogarithms, polylogarithms, symbolic integration, computer algebra, Feynman integrals, $\varepsilon$-expansions
\\
{\em Classification:}
	4.4 Feynman diagrams, 5 Computer Algebra
\\
{\em External routines/libraries:}
\\
{\em Subprograms used:}
\\
{\em Nature of problem:}
	Feynman integrals and their $\varepsilon$-expansions in dimensional regularization can be expressed in the Schwinger parametrization as multi-dimensional integrals of rational functions and logarithms. 
	Symbolic integration of such functions therefore serves a tool for the exact and direct evaluation of Feynman graphs.
\\
{\em Solution method:}
	Symbolic integration of rational linear combinations of polylogarithms of rational arguments is obtained using a representation in terms of hyperlogarithms. The algorithms exploit their iterated integral structure.
\\
{\em Restrictions:}
	To compute multi-dimensional integrals with this method, the integrand must be linearly reducible, a criterion we state in section~\ref{sec:polynomial-reduction}.
	As a consequence, only a small subset of all Feynman integrals can be addressed.
\\
{\em Unusual features:}
	The complete program works strictly symbolically and the obtained results are exact.
	Whenever a Feynman graph is linearly reducible, its $\varepsilon$-expansion can be computed to arbitrary order (subject only to time and memory restrictions) in $\varepsilon$, near any even dimension of space-time and for arbitrarily $\varepsilon$-dependent powers of propagators with integer values at $\varepsilon=0$.
	Also the method is not restricted to scalar integrals only, but arbitrary tensor integrals can be computed directly.
\\
{\em Additional comments:}
	Further applications to parametric integrals, outside the application to Feynman integrals in the Schwinger parametrization, are very likely.
\\
{\em Running time:}
	Highly dependent on the particular problem through the number of integrations to be performed (edges of a graph), the number of remaining variables (kinematic invariants), the order in $\varepsilon$ and the complexity of the geometry (topology of the graph). Simplest examples finish in seconds, but the time needed increases beyond any bound for sufficiently high orders in $\varepsilon$ or graphs with many edges.
\\

\end{small}

\section{Introduction}
\label{sec:introduction}
An important class of special functions is given by multiple polylogarithms \cite{Goncharov:MplCyclotomyModularComplexes,BorweinBradleyBroadhurstLisonek:SpecialValues} 
\begin{equation}
	\label{eq:def:Li}%
	\Li_{n_1,\ldots,n_r} (z_1,\ldots,z_r)
	\defas
	\sum_{0<k_1<\cdots<k_r}
	\frac{z_1^{k_1} \cdots z_r^{k_r}}{k_1^{n_1} \cdots k_r^{n_r}}
\end{equation}
of several complex variables $\vec{z}$, which generalize the traditional polylogarithms $\Li_n(z)$ of a single variable (the case $r=1$) studied for example in \cite{Lewin:PolylogarithmsAssociatedFunctions}. Many properties and relations of these multivalued functions can be formulated and studied conveniently in terms of combinatorial structures, which renders them suitable for symbolic algorithms that can be implemented on a computer.

This is mainly a consequence of their representation as a special class of iterated integrals \cite{Chen:II} and our preferred basis are the classic hyperlogarithms \cite{LappoDanilevsky} of
\begin{definition}%
	\label{def:hyperlog} %
	Given a finite set $\Sigma \subset \C$ containing $0\in \Sigma$, each word $w = \letter{\sigma_1}\ldots\letter{\sigma_n} \in \Sigma^{\times}$ ($\letter{\sigma}$ denotes the letter for $\sigma \in \Sigma$) defines the \emph{hyperlogarithm} $\Hyper{w}$ by setting $\Hyper{\letter{0}^n}(z) \defas \frac{\log^n z}{n!}$ and otherwise recursively applying
	\begin{equation}%
		\label{eq:def:hyperlog} %
		\Hyper{\letter{\sigma}w'}
		\defas
		\int_0^z \frac{\dd z'}{z'-\sigma} \Hyper{w'}(z').
	\end{equation} %
	We also abbreviate $\Hyper{\sigma_1, \ldots, \sigma_n} \defas \Hyper{\letter{\sigma_1}, \ldots, \letter{\sigma_n}}$ and write $\sigma^{(n)}$ for a sequence $\sigma,\ldots,\sigma$ of $n$ letters $\sigma$.
\hide{
	\begin{equation}
		\label{eq:def:hyperlog} %
		\Hyper{w}(z)
		\defas
		\begin{cases}
			\frac{\log^n z}{n!} & w =\letter{0}^n \\
			\int_0^z \frac{\dd z'}{z'-\sigma} \Hyper{w'}(z') & \text{otherwise, $w=\letter{\sigma} w'$} \\
		\end{cases} %
	\end{equation} %
} %
\end{definition}
These functions are also referred to as generalized harmonic polylogarithms (with linear weights) or Goncharov polylogarithms, since they relate to \eqref{eq:def:Li} via
\begin{equation}
	(-1)^r
	\Li_{\vec{n}} \left( \vec{z} \right)
	=
	\Hyper{0^{(n_r -1)},\sigma_r,\ldots,0^{(n_2 - 1)},\sigma_2,0^{(n_1-1)},\sigma_1}(z)
	\label{eq:Hyper-as-Li} %
\end{equation}
where $\vec{n} = (n_1,\ldots,n_r) \in \N^r$ and $\sigma_1,\ldots,\sigma_r \neq 0$ are such that
$
 \vec{z}
 =
	\left(
 		\frac{\sigma_2}{\sigma_1},
 		\frac{\sigma_3}{\sigma_2},
		\ldots,
		\frac{z}{\sigma_r}
	\right)
$,
equivalently 
$
	\sigma_i 
	= 
	z \prod_{k=i}^{r} z_k^{-1}
$.
Particle physicists observed special classes of hyperlogarithms in results of Feynman integral calculations. The most famous example is the case when $\Sigma \subseteq \set{-1,0,1}$, called harmonic polylogarithms in \cite{RemiddiVermaseren:HarmonicPolylogarithms}, and practical tools to compute with these are available like \cite{Maitre:HPL,Maitre:HPLComplex}. Some algorithms for general hyperlogarithms are also implemented in \cite{AblingerBluemleinSchneider:GeneralizedHarmonicSumsAndPolylogarithms}. However, the full power of definition~\ref{def:hyperlog} can be used not only to express the result of Feynman integrals, but actually to compute them in the first place.

Namely, the study \cite{Brown:MZVPeriodsModuliSpaces} of periods of moduli spaces of curves of genus zero computed multiple integrals
\begin{equation}
	f_k
	\defas
	\int_0^{\infty} 
	f_{k-1}(z_k)
	\ \dd z_k
	=
	\int_{0}^{\infty} \dd z_1
		\cdots
		\int_0^{\infty} \dd z_k
	\ f_0
	\label{eq:iteration-partial-integrals} %
\end{equation}
of certain polylogarithms $f_0(\vec{z})$ such that each of the partial integrals $f_k$ is a hyperlogarithm in the next integration variable $z_{k+1}$. This criterion on $f_0$ is called \emph{linear reducibility} in \cite{Brown:TwoPoint}, where the symbolic integration algorithm of such functions is explained and applied theoretically to some finite scalar single-scale Feynman integrals (massless propagators). In \cite{Brown:PeriodsFeynmanIntegrals} it was further shown that linear reducibility is actually fulfilled for an infinite family of non-trivial Feynman integrals, but still explicit results were missing.

This technique has then practically been used in \cite{ChavezDuhr:Triangles} to compute off-shell three-point functions and in \cite{Wissbrock:Massive3loopLadder,Wissbrock:New3loopHeavyFlavor,AblingerBluemleinRaabSchneiderWissbrock:Hyperlogarithms} to calculate operator insertions into propagator graphs containing a single non-zero mass scale. A further application to phase-space integrals related to Higgs production can be found in \cite{AnastasiouDuhrDulatMistlberger:SoftTripleRealHiggs}.

Unfortunately, none of these programs was made publicly available so far. This might partly be due to the fact that the exposition in \cite{Brown:TwoPoint} does not provide a simple method to obtain certain integration constants in a crucial intermediate step of the algorithm.
In fact, \cite{AnastasiouDuhrDulatMistlberger:SoftTripleRealHiggs} resorts to numeric evaluations to guess these constants and a similar approach is common to many applications of the symbol- and coproduct-calculus \cite{Duhr:HopfAlgebrasCoproductsSymbols,GoncharovSpradlinVerguVolovich:ClassicalPolylogarithmsAmplitudesWilsonLoops,DuhrGanglRhodes:PolygonsAndSymbols}. Also within the method of differential equations \cite{ArgeriMastrolia:FeynmanDiagramsDifferentialEquations}, boundary conditions occur that must be obtained separately, e.g\ through physical reasoning or separate computations of expansions in certain limits.

We close this gap and provide a complete implementation of the method \cite{Brown:TwoPoint} of symbolic integration using hyperlogarithms in the computer algebra system {\Maple} \cite{Maple}. This program was used in \cite{Panzer:MasslessPropagators,Panzer:DivergencesManyScales} to compute several non-trivial Feynman integrals (including divergences and complicated kinematics) and we hope that it will prove helpful in further applications by physicists and mathematicians alike.

Since our foremost goal was to supply a tool for the computation of Feynman integrals, we did not aim for a most general computer algebra framework to handle hyperlogarithms but instead focussed on this particular application. Still, the algorithms were implemented for very general situations and may be used for different problems as well.

For completeness let us mention that while we focus on polylogarithms as iterated integrals, the representation \eqref{eq:def:Li} as nested sums opens the door to completely different strategies like \cite{MochUwerWeinzierl:NestedSums} with implementations readily available \cite{Weinzierl:SymbolicExpansion,MochUwer:XSummer}. A lot of progress is being made on symbolic manipulation of sums and we like to point out \cite{AblingerBluemleinSchneider:GeneralizedHarmonicSumsAndPolylogarithms} and the numerous references therein.
However, we will not pursue this approach in our work.

\subsection{Plan of the paper}
\label{sec:plan-of-paper}%
In section~\ref{sec:algorithms} we present our algorithms to symbolically manipulate hyperlogarithms in sufficient detail so as to make an implementation straightforward. We follow the ideas of \cite{Brown:TwoPoint} where the reader might find illuminating examples and details. Our main original contribution is section~\ref{sec:reglim-algorithm} where we solve the problem of determination of integration constants mentioned above.

The {\Maple} implementation {\HyperProg} is presented in section~\ref{sec:implementation} and includes examples of its application to integration problems and for transformations of arguments of polylogarithms.

To apply these methods to multiple integrals \eqref{eq:iteration-partial-integrals}, we review the property of linear reducibility in section~\ref{sec:polynomial-reduction} and explain how to exploit the polynomial reduction algorithm contained in {\HyperProg}.

Section~\ref{sec:feynman-integrals} is devoted to our original motivation and main application: the calculation of Feynman integrals. In {\HyperProg} we supply a couple of commands to facilitate the work with Feynman graphs. Detailed examples and demonstrations are contained in the attached {\Maple} worksheet \Filename{Manual.mw}.

To ensure correctness of our program, we performed plenty of tests. Some of them are summarized in \ref{sec:tests} and provided in the file \Filename{HyperTests.mpl}.

Some combinatorial proofs were delegated to \ref{sec:proofs} and we supply a short reference of functions and options provided by {\HyperProg} in \ref{sec:reference}.

\section{Algorithms for hyperlogarithms}
\label{sec:algorithms}
We already mentioned references on hyperlogarithms, multiple polylogarithms and iterated integrals. In section~\ref{sec:tensor-algebra-iterated-integrals} we collect some standard results and fix our notations.

Afterwards we follow the ideas of \cite{Brown:TwoPoint} for the integration of hyperlogarithms and explain in detail how each step can be implemented combinatorially. In short, to compute $f_{k} = \int_0^{\infty} \dd z_{k} f_{k-1}(z_k)$ for some polylogarithm $f_{k-1}$, we proceed in three steps:
\begin{enumerate}
	\item
		Express $f_{k-1}$ as a hyperlogarithm in $z_k$.
	\item
		Find a primitive $F_{k-1}(z_k)$ such that $\partial_{z_k} F_{k-1} = f_{k-1}$.
	\item
		Evaluate the limits $f_k = \lim_{z_k \rightarrow \infty} F_k - \lim_{z_k \rightarrow 0} F_k$.
\end{enumerate}
Many of these operations are straightforward or explained with examples in \cite{Brown:TwoPoint}. All the work actually lies in step 1. above, which is the subject of the reviewing section~\ref{sec:reginf-as-hyperlog} and our additional algorithm of section~\ref{sec:reglim-algorithm} to symbolically compute constants of integration.

The original work \cite{Brown:MZVPeriodsModuliSpaces} in the setting of moduli spaces contains a lot of details, worked examples and geometric interpretations of the ideas employed. In particular we recommend sections 5 and 6 therein which develop the theory of hyperlogarithms tailored to our setup, including logarithmic regularization in detail.

Let us remark that instead of tracking a sequence \eqref{eq:iteration-partial-integrals} of one-dimensional iterated integrals, the natural approach would be to consider every $f_k$ as an iterated integral in several variables $z_{k+1}, z_{k+2}, \ldots$ simultaneously. This idea is pursued in \cite{BognerBrown:SymbolicIntegration} and their authors are currently finalizing an implementation of this method as well.

\subsection{The tensor algebra and iterated integrals}
\label{sec:tensor-algebra-iterated-integrals}%
The algebraic avatar of iterated integrals is the shuffle algebra
\begin{equation}
	T(\Sigma)
	\defas
	\bigoplus_{w \in \Sigma^{\times}} \Q w
	=
	\bigoplus_{n=0}^{\infty}
	T_n(\Sigma)
	,\ 
	T_n(\Sigma)
	\defas
	\left(\Q\Sigma \right)^{\tp n}
	\label{eq:def:shuffle-algebra} %
\end{equation}
spanned by all words over the alphabet $\Sigma$; some references for this Hopf algebra are \cite{Reutenauer:FreeLieAlgebras,Sweedler}. It is graded by the weight $n = \abs{w}$ counting the number of letters in a word ($w = \letter{\sigma_1}\ldots\letter{\sigma_n} \in \Sigma^n$).
Apart from the non-commutative concatenation product, it is equipped with the commutative shuffle product defined recursively by
\begin{equation}
	\label{eq:shuffle-product} %
	(\letter{\sigma}w) \shuffle (\letter{\tau} w')
	\defas
		\letter{\sigma}(w \shuffle \letter{\tau} w')
	+ \letter{\tau} (\letter{\sigma}w \shuffle w')
\end{equation}
until $\emptyWord \shuffle w = w \shuffle \emptyWord = w$ where $\emptyWord$ denotes the empty word which is the identity of $T(\Sigma)$. The coproduct $\Delta: T(\Sigma) \longrightarrow T(\Sigma) \tp T(\Sigma)$ of interest is the deconcatenation
\begin{equation}
	\label{eq:def:deconcatenation} %
	\Delta\left( \letter{\sigma_1}\ldots\letter{\sigma_n} \right)
	\defas
	\sum_{k=0}^{n} \letter{\sigma_1}\ldots\letter{\sigma_k} \tp \letter{\sigma_{k+1}} \ldots \letter{\sigma_n}.
\end{equation}
These combinatorial structures precisely capture the analytic properties of the iterated integrals \cite{Chen:II,Hain:GeometryMHSFundamentalGroup}
\begin{equation}
	\iInt_{\gamma} \letter{\sigma_1}\ldots\letter{\sigma_n}
	\defas
	\int_{0}^{1}
	\frac{\dd \gamma(z)}{\gamma(z) - \sigma_1}
	\iInt_{\restrict{\gamma}{[0,z]}} \letter{\sigma_2}\ldots\letter{\sigma_n}
	\label{eq:def:iterated-integral} %
\end{equation}
of differential one-forms $\letter{\sigma} = \frac{\dd z}{z - \sigma} \in \Omega^1(\C \setminus \Sigma)$ of a single variable $z$. With $\iInt_{\gamma} \emptyWord \defas 1$, \eqref{eq:def:iterated-integral} defines homotopy invariant functions of a path $\gamma: [0,1] \rightarrow \C \setminus \Sigma$ of integration which are often identified with the resulting multivalued analytic functions of the endpoint $z = \gamma(1) \in \C \setminus \Sigma$.

Observe that in \eqref{eq:def:hyperlog} we chose the singular base point $\gamma(0) = 0 \in \Sigma$ which is the reason why we had to define $\Hyper{\letter{0}}(z) \defas \log z$ specially and not by the divergent integral $\int_0^{z} \frac{\dd z}{z}$.

We extend the maps $w \mapsto \iInt_{\gamma} w$ and $w \mapsto \Hyper{w}$ linearly to the whole shuffle algebra $T(\Sigma)$. Then the fundamental properties of iterated integrals become
\begin{enumerate}
	\item $\iInt_{\gamma} w \cdot \iInt_{\gamma} w' = \iInt_{\gamma} (w \shuffle w')$, i.e. $\iInt_{\gamma}$ is multiplicative,

	\item By Chen's lemma, concatenation $\gamma \concat \eta$ of two paths with $\eta(1) = \gamma(0)$ gives for every word $w = \letter{\sigma_1} \ldots \letter{\sigma_n}$
	\begin{equation}
		\label{eq:Chens-lemma} %
			\iInt_{\gamma \concat \eta} w
			= \sum_{i=0}^{n} 
					\iInt_{\gamma} \letter{\sigma_1}\ldots\letter{\sigma_i}
					\cdot 
					\iInt_{\eta} \letter{\sigma_{i+1}}\ldots\letter{\sigma_n}
			.
	\end{equation}

	\item $\Q(z) \tp_{\Q} T(\Sigma) \ni f \tp w \mapsto \left[ \gamma \mapsto f(\gamma(1)) \cdot \iInt_{\gamma} w \right]$ is injective, so iterated integrals associated to different words are linearly independent with respect to rational (actually even for algebraic) prefactors $f$.
\end{enumerate}
The analogous properties hold for the hyperlogarithms $w \mapsto \Hyper{w}$ of \eqref{eq:def:hyperlog}. These functions $\Hyper{w}: \C\setminus \Sigma \longrightarrow \C$ are single-valued once we restrict to the simply connected domain where $0<\abs{z}<\min\setexp{\abs{\sigma}}{0\neq\sigma \in \Sigma}$ and $z \notin (-\infty,0]$, after fixing $\log$ to the principal branch with $\log 1 = 0$. In the sequel we will only consider such hyperlogarithms $f(z) = \Hyper{w}(z)$ that allow for an analytic continuation to all of $(0,\infty)$. This is necessary to give the integrals $\int_0^{\infty} f(z)\ \dd z$ we want to compute a well-defined value.

\subsection{Integration and differentiation}
\label{sec:integration-differentiation}%
We consider the algebra 
$
	L(\Sigma)
	\defas
	\regulars_{\Sigma} \left[ 
		\Hyper{w}(z): w \in \Sigma^{\times}
	\right]
$
spanned by hyperlogarithms with rational prefactors whose denominators factor linearly with zeros in $\Sigma$ only:
\begin{equation}
	\regulars_{\Sigma}
	\defas
	\Q\Big[
		z, 
		\frac{1}{z-\sigma}: \sigma \in \Sigma
	\Big].
	\label{eq:def:regulars}
\end{equation}
By construction we have $\partial_z \Hyper{\letter{\sigma_1}\ldots\letter{\sigma_n}}(z) = \frac{1}{z-\sigma_1} \Hyper{\letter{\sigma_2}\ldots\letter{\sigma_n}}(z)$ such that $L(\Sigma)$ is closed under $\partial_z$, while for any $f \in L(\Sigma)$ we can find primitives $F \in L^{+} (\Sigma)$, $\partial_z F(z) = f(z)$, in the enlarged algebra
$
	L^{+}(\Sigma)
	\defas
	\regulars_{\Sigma}^{+} \left[
		\setexp{\Hyper{w}}{w\in\Sigma^{\times}}
	\right]
$
where
\begin{equation}
	\regulars_{\Sigma}^{+}
	\defas
	\regulars_{\Sigma}\left[
			\Sigma 
			\cup
			\setexp{\frac{1}{\sigma-\tau}}{\sigma,\tau\in\Sigma \ \text{and}\ \sigma\neq\tau} 
	\right].
	\label{eq:def:L+}
\end{equation}
Namely, a primitive for $g(z) \Hyper{w}(z)$ can be constructed by partial fractioning the rational prefactor
\begin{equation}
	g(z)
	=
		\sum_{\sigma\in\Sigma} \sum_{n\in\N} \frac{A_{\sigma,n}}{(z-\sigma)^n} 
		+ \sum_{n\in\N_0} A_n z^n
	\in
	\regulars_{\Sigma},
\end{equation}
setting $F=\Hyper{\letter{\sigma}w}(z)$ as a primitive of $\frac{\Hyper{w}(z)}{z-\sigma}$ and repeated use of the partial integration formulae 
\begin{align}
	\int \frac{\dd z\ \Hyper{w}(z)}{(z-\sigma)^{n+1}} 
	& =
		-\frac{\Hyper{w}(z)}{n(z-\sigma)^n} 
		+ \int \frac{\dd z\ \partial_z \Hyper{w}(z)}{n(z-\sigma)^n} ,
	\\
	\int \dd z\ z^n \Hyper{w}(z)
	&=
		\frac{z^{n+1}\cdot\Hyper{w}(z)}{n+1} 
		- \int \frac{\dd z\ z^n}{n+1} \partial_z \Hyper{w}(z)
	\label{}
\end{align}
to reduce the problem of finding a primitive to the case where the hyperlogarithm $\partial_z \Hyper{\letter{\sigma_1}\ldots\letter{\sigma_n}}(z) = \frac{\Hyper{\letter{\sigma_2}\ldots\letter{\sigma_n}}(z)}{z-\sigma}$ is of lower weight. This recursion terminates when $w$ becomes the empty word. Hence computation of a convergent integral $\int_0^{\infty} f(z) \dd z$ for $f\in L(\Sigma)$ reduces to obtaining a primitive $F \in L^+(\Sigma)$ of $f$ as described and evaluating the limits
\begin{equation}
	\int_0^{\infty} f(z)\ \dd z
	= \lim_{z \rightarrow \infty} F(z)
		- \lim_{z \rightarrow 0} F(z).
	\label{eq:integral-as-difference-of-limits} %
\end{equation}

\subsection{Divergences and logarithmic regularization}
\label{sec:divergences-regularization} %
	The singularities of $\Hyper{w}(z)$ at $z\rightarrow \tau \in \Sigma \cup \set{\infty}$ are at worst logarithmic, namely for any $w \in T(\Sigma)$ there is a decomposition
\begin{equation}
	\Hyper{w}(z)
	=
	\sum_{i=0}^{\abs{w}}
	f_{w,\tau}^{(i)}(z)
	\cdot
	\begin{cases}
		\log^i z, & \tau=\infty\\
		\log^i (z-\tau), & \tau\neq \infty\\
	\end{cases}
	\label{eq:regularization}
\end{equation}
with functions $f_{w,\tau}^{(i)}(z)$ uniquely defined upon the requirement of being holomorphic at $z \rightarrow \tau$; for $t=\infty$ this means holomorphy of $f_{w,\infty}^{(i)}\left( \frac{1}{z} \right)$ at $z\rightarrow 0$.
Note that $\Hyper{w}(z)$ is finite for $z\rightarrow\tau \notin \set{0,\infty}$ whenever $w$ does not begin with the letter $\letter{\tau}$.

The \emph{regularized limits} are defined for any $\tau$ as
\begin{equation}
	\AnaReg{z}{\tau} \Hyper{w}(z)
	\defas
	f_{w,\tau}^{(0)} (\tau),
	\label{eq:def:reglim}
\end{equation}
	such that $\lim_{z \rightarrow \tau} \Hyper{w}(z) = \AnaReg{z}{\tau} \Hyper{w}(z)$ whenever this limit is finite. 
	The advantage is then that by linearity,
	\begin{equation*}
		\lim_{z \rightarrow \tau}
		f(z)
		= \sum_{w \in \Sigma^{\times}} \lambda_{w} \AnaReg{z}{\tau} \Hyper{w}(z)
		\ \text{for}\ 
		f(z) 
		= 
			\sum_{w \in \Sigma^{\times}} \lambda_w \Hyper{w}(z)
	\end{equation*}
	can be computed for each word $w$ separately and is thus well suited for an implementation, even though the limits $\lim_{z \rightarrow \tau}\Hyper{w}(z)$ might diverge individually.
\begin{definition}
		For disjoint sets $A,B \subset \Sigma$ the projection
$
	\WordReg{A}{B}: T(\Sigma) \longrightarrow T(\Sigma)
$
	is determined by the requirements
	\begin{enumerate}
		\item $\WordReg{A}{B} (w \shuffle w') = \WordReg{A}{B}(w) \shuffle \WordReg{A}{B} (w')\ \forall w,w' \in \Sigma^{\times}$,

		\item $\WordReg{A}{B}(w) = w = \letter{\sigma_1}\ldots\letter{\sigma_n}$ if $\sigma_1 \notin B$ and $\sigma_n \notin A$,

		\item $\WordReg{A}{B}(w) = 0$ for all $1 \neq w \in A^{\times} \cup B^{\times}$.
	\end{enumerate}
	We write $\WordReg{\sigma}{\tau} \defas \WordReg{\set{\sigma}}{\set{\tau}}$ and suppress empty sets in the notation, e.g. $\WordReg{\sigma}{} = \WordReg{\sigma}{\emptyset}$ and $\WordReg{}{\tau} = \WordReg{\emptyset}{\tau}$.
\end{definition}
This \emph{shuffle-regularization} is a combinatorial operation that projects onto words that neither begin with a letter in $A$ nor end with a letter from $B$. 
Every word $w \in T(\Sigma)$ decomposes uniquely as
\begin{equation}
	w 
	=
		\sum_{a \in A^{\times}} 
		\sum_{b \in B^{\times}} 
			a \shuffle b \shuffle w_{A,B}^{(a,b)}
	\label{eq:shuffle-decomposition} %
\end{equation}
into such \emph{$A$-$B$-regularized} words $w_{A,B}^{(a,b)} \in \im \WordReg{A}{B}$ and thus $\WordReg{A}{B}(w) = w_{A,B}^{(1,1)}$. To compute \eqref{eq:shuffle-decomposition} we can use
\begin{lemma}
	\label{lemma:shuffle-head-tail-decomposition} %
	For $w = u \letter{\sigma} a$ with $a = \letter{a_1}\ldots\letter{a_n}$,
\begin{equation}
	w
	=
	\sum_{i=0}^n
		\left[	u
						\shuffle 
						(-\letter{a_i})\ldots(-\letter{a_1}) 
		\right]
		\letter{\sigma}
		\shuffle
		\letter{a_{i+1}}\ldots\letter{a_n}.
	\label{eq:shuffle-tail-decomposition} %
\end{equation}
So when $\sigma\notin A$ is the last letter of $w$ not in $A$, thus $a \in A^{\times}$, we deduce $\WordReg{A}{}(w) = (-1)^n \left( u \shuffle \letter{a_n}\ldots\letter{a_1} \right) \letter{\sigma}$.

Analogously $\WordReg{}{B} (b \letter{\sigma}u) = \letter{\sigma} \left( u \shuffle S (b) \right)$ for $b \in B^{\times}$ and $\sigma\notin B$, setting $S(b_1\ldots b_k) \defas (-b_k)\ldots(-b_1)$.

Finally note $\WordReg{A}{B} = \WordReg{A}{} \circ \WordReg{}{B} = \WordReg{}{B} \circ \WordReg{A}{}$.
\end{lemma}

For $A=\set{0}$ and $B=\emptyset$, \eqref{eq:shuffle-decomposition} reads
$
	w = \sum_{i} \letter{0}^i \shuffle w_i
$
where $w_i$ do not end in $\letter{0}$. Since $\Hyper{w_i}(z)$ is holomorphic at $z\rightarrow 0$ and
$
	\Hyper{w}(z)
	= \sum_{i} \log^i z \cdot \Hyper{w_i} (z)
$
reveal $f_{0,w}^{i}(z) = \Hyper{w_i} (z)$ from \eqref{eq:regularization}, we can compute the limit
\begin{equation}
	\AnaReg{z}{0} \Hyper{w}(z)
	= \Hyper{\WordReg{0}{}(w)} (0)
	= 0
	\ \text{when}\ 
	w \in \Sigma^{\times} \setminus \set{1}.
	\label{eq:reg0}
\end{equation}
In fact our definition~\eqref{eq:def:hyperlog} is deliberately tuned such that the empty word $w=1 \mapsto \Hyper{w}(z) = 1$ is the only word in $\Sigma^{\times}$ with non-vanishing $\AnaReg{z}{0} \Hyper{w}(z)$.

\begin{lemma}
	\label{lemma:zero-expansion} %
	Let $w = \sum_i \letter{0}^i \shuffle w_i \in L(\Sigma)$ for $\WordReg{0}{}(w_i) = w_i$ not ending on $\letter{0}$. Then $\Hyper{w}(z) = \sum_i \frac{\log^i z}{i!} \Hyper{w_i}(z)$ and $\Hyper{w_i}(z)$ are holomorphic at $z \rightarrow 0$ and their series expansion $\Hyper{w_i}(z) = \sum_{n\geq 0} a_n z^n$ can be directly computed (recursively) from the iterated integral representation: Starting with the empty word $\Hyper{1}(z) = 1$, let $\Hyper{w}(z) = \sum_{n\geq 0} a_n z^n$. Then
	\begin{align}
		\Hyper{\letter{0}w}(z)
		&= 
		\sum_{n =1}^{\infty} \frac{a_n}{n} z^n
		\quad\text{and for any}\quad
		\sigma \in \Sigma \setminus\set{0},
		\\
		\Hyper{\letter{\sigma}w}(z)
		&=
		\frac{1}{-\sigma} \sum_{n,m = 0}^{\infty} \frac{a_n}{\sigma^m (n+m+1)} z^{n+m+1}.
		\label{eq:zero-expansion-recursive} %
	\end{align}
\end{lemma}
For expansions up at infinity, we first introduce an intermediary point $u \in (0,\infty)$ to split up the integration using Chen's lemma \eqref{eq:Chens-lemma}, and then let $u \rightarrow \infty$:
\begin{equation}
	\Hyper{w}(z)
	= 
	\sum_{k=0}^n
	\AnaReg{u}{\infty} \iInt_{u}^{z} \letter{\sigma_1}\ldots\letter{\sigma_k}
	\cdot
	\AnaReg{u}{\infty} \Hyper{\letter{\sigma_{k+1}}\ldots\letter{\sigma_n}} (u).
	\label{eq:inf-expansion-deconcatenation} %
\end{equation}
\begin{definition}
	For a word $w=\letter{\sigma_1}\ldots\letter{\sigma_n}$, let
	\begin{equation}
		\WordReg{}{\infty}(w)
		\defas
			\sum_{k=1}^{n}
			\left( \letter{\sigma_{k}} - \letter{-1} \right)
			\left[ 
			(-\letter{-1})^{k-1}
				\shuffle
				\letter{\sigma_{k+1}}\ldots\letter{\sigma_n}
			\right]
		\label{eq:def:word-reginf} %
	\end{equation}
	denote the projection of $T(\Sigma)$ on words beginning with differences $(\letter{\sigma}-\letter{-1})$ that annihilates $\WordReg{}{\infty} (\letter{-1}^n) = 0$ for any $n>0$. Further set $\WordReg{0}{\infty} \defas \WordReg{}{\infty} \circ \WordReg{0}{}$.
\end{definition}
If 
$
	w
	=
	(\letter{\sigma}-\letter{-1}) w'
	\in \im (\WordReg{0}{\infty})
$, then
\begin{equation}
	\Hyper{w}(z)
	= \int_0^{z} \frac{(1+\sigma)\dd z'}{(z' - \sigma)(z'+1)}
		\Hyper{w'}(z')
	\label{eq:reginf-convergent-integral} %
\end{equation}
reveals that $\Hyper{w}(\infty) = \AnaReg{z}{\infty} \Hyper{w}(z)$ is finite as an absolutely convergent integral since $\Hyper{w'}(z')$ grows at worst logarithmically by \eqref{eq:regularization}. 
We therefore conclude
\begin{equation}
	\AnaReg{z}{\infty} \Hyper{w}(z)
	=
	\Hyper{\WordReg{0}{\infty}(w)} (\infty)
	\ \text{for any}\ 
	w \in T(\Sigma)
	\label{eq:reginf} %
\end{equation}
 from $\displaystyle\AnaReg{z}{\infty} \Hyper{\letter{-1}}(z) = 0 = \AnaReg{z}{\infty} \Hyper{\letter{0}}$ and
\begin{lemma}
	\label{lemma:word-reginf-decomposition} %
	For any $w,w' \in L(\Sigma)$, $\WordReg{}{\infty}(w\shuffle w') = \WordReg{}{\infty}(w) \shuffle \WordReg{}{\infty}(w')$ is multiplicative and for any word $w=\letter{\sigma_1}\ldots\letter{\sigma_n}$, the following identity holds:
\begin{equation}
	w
	=
		\sum_{k=0}^n \letter{-1}^{k}
		\shuffle
		\WordReg{}{\infty} \left( 
			\letter{\sigma_{k+1}}\ldots\letter{\sigma_n}
		\right).
	\label{eq:word-reginf-decomposition} %
\end{equation}
\end{lemma}
The second ingredient to compute \eqref{eq:inf-expansion-deconcatenation} lies in
\begin{lemma}
	\label{lemma:moebius-transform} %
	For any M\"{o}bius transform $f(z)=\frac{az+b}{cz+d}$,
\begin{equation}
	\frac{\dd f^{-1}(z)}{f^{-1}(z) - \sigma}
	=
	\frac{\dd z}{z - f(\sigma)} - \frac{\dd z}{z-f(\infty)}
	\label{eq:moebius-differentials} %
\end{equation}
wherefore $\iInt_{A}^{B} w = \iInt_{f(A)}^{f(B)} \WordTransformation{f}(w)$ with the linear (and multiplicative) map $\WordTransformation{f}$ that replaces any letter $\letter{\sigma}$ by
\begin{equation}
	\WordTransformation{f} (\letter{\sigma})
	\defas
	\letter{f(z)} - \letter{f(\infty)},
	\ \text{dropping any}\ 
	\letter{\infty} \defas 0.
	\label{eq:moebius-transform} %
\end{equation}
\end{lemma}
We apply this to
$
	\iInt_{u}^{z} w
	= \iInt_{1/u}^{1/z} \WordTransformation{\frac{1}{z}} (w)
$
and recall that $\lim_{u\rightarrow \infty} \iInt_{1/u}^{1/z} w = \Hyper{w}(\frac{1}{z})$ is finite for $w \in \im \left( \WordReg{0}{} \right)$ not ending in $\letter{0}$. Furthermore,
\begin{equation*}
	\AnaReg{u}{\infty} \iInt_{1/u}^{1/z} \letter{0}
	=
	\AnaReg{u}{\infty} \log \frac{u}{z}
	=
	\log\frac{1}{z}
	=
	\Hyper{\letter{0}}\left( \frac{1}{z} \right)
\end{equation*}
completes \eqref{eq:inf-expansion-deconcatenation} to a combinatorial algorithm of the form
\begin{equation}
	\Hyper{w}(z)
	= \sum_{w} \Hyper{\WordTransformation{\frac{1}{z}}(w_1)}\left( \frac{1}{z} \right)
		\cdot \Hyper{\WordReg{0}{\infty}\left( w_2 \right)}(\infty).
	\label{eq:inf-expansion}
\end{equation}
Employing \eqref{eq:zero-expansion-recursive}, this equation can be used to expand $\Hyper{w}(z)$ at $z\rightarrow \infty$ as a polynomial in $\log z$ and a power series in $\frac{1}{z}$. This suffices to compute $\lim_{z \rightarrow \infty} F(z)$ in \eqref{eq:integral-as-difference-of-limits}.

\subsection{Regularized limits as hyperlogarithms}
\label{sec:reginf-as-hyperlog}%
When we follow \eqref{eq:iteration-partial-integrals}, after taking the limits \eqref{eq:integral-as-difference-of-limits} we will from \eqref{eq:inf-expansion} have a representation of the partial integral $F_k$ in terms of expressions $\Hyper{\WordReg{0}{\infty}(w)}(\infty)$ that depend on the next integration variable $ t \defas z_{k+1}$ implicitly through the letters in the word $w$. To proceed with the integration process, we must rewrite $F_k$ as a hyperlogarithm in $t$.

So let $w=\letter{\sigma_1}\ldots\letter{\sigma_n}$ ($\sigma_n \neq 0$) with letters $\sigma_i(t)$ depending on a parameter $t$, then we can take the derivative $\partial_t \Hyper{w}(z)$ in the integrand of the iterated integral $\Hyper{w}$. Partial fractioning and partial integration suffice to prove 
	\begin{align*}
		\partial_t \Hyper{w}(z)
		&=
		\sum_{i=1}^{n-1}
			\left[
				\frac{
					\partial_t (\sigma_i(t) - \sigma_{i+1}(t))
				}{
					\sigma_i(t)-\sigma_{i+1}(t)
				}
			\right]
			\Hyper{
					\ldots\not\letter{\sigma_{i+1}}\ldots
				- \ldots \not\letter{\sigma_i}\ldots
			}(z)
		\\ &
		+
			\left[
				\frac{-\partial_t \sigma_1}{z - \sigma_1}
			\right]
			\Hyper{\letter{\sigma_2}\ldots\letter{\sigma_n}}(z)
		- \left[
				\frac{\partial_t \sigma_n}{\sigma_n}
			\right]
			\Hyper{\letter{\sigma_1}\ldots\letter{\sigma_{n-1}}}(z)
	\end{align*}
	where $\not\letter{\sigma_i}$ means to delete the letter $\letter{\sigma_i}$ from $w$. Applying $\AnaReg{z}{\infty}$ and exploiting $\AnaReg{z}{\infty} \partial_t = \partial_t \AnaReg{z}{\infty}$ yields
	\begin{align}
		&\partial_t \AnaReg{z}{\infty} \Hyper{w}(z)
		=
		-\left[ \partial_t \ln \sigma_n(t) \right]
		\cdot
		\AnaReg{z}{\infty}
		\Hyper{\ldots\not\letter{\sigma_n}}(z)
		\label{eq:diff-reginf}%
		\\ 
		&\quad+
		\sum_{i=1}^{n-1}
		\left[ \partial_{t} \ln (\sigma_i - \sigma_{i+1} ) \right]
		\cdot
		\AnaReg{z}{\infty}
		\Hyper{\ldots\not\letter{\sigma_{i+1}}\ldots-\ldots\not\letter{\sigma_i}\ldots}(z). 
		\nonumber %
	\end{align}
We assume that $\sigma_1,\ldots,\sigma_n \in \Q(t)$ are rational, such that any
$
	\sigma_i(t)-\sigma_j(t) 
	= c \prod_{\tau} (t - \tau)^{\lambda_{\tau}}
$
factors linearly\footnote{with zeros $\tau \in \overline{\Q}$ in the algebraic closure and integer multiplicities $\lambda \in \Z$} and thus
$
	\partial_t \ln \left[ \sigma_i(t) - \sigma_j(t) \right]
	= \sum_{\tau} \frac{\lambda_{\tau}}{t-\tau}
$
together with \eqref{eq:diff-reginf} prove
\begin{lemma}
	\label{lemma:reginf-is-hyperlog}%
	For rational letters $\Sigma = \set{\sigma_1(t), \ldots, \sigma_N(t)} \subset \Q(t) \setminus (0,\infty)$ without positive constants (such that the limit $\AnaReg{z}{\infty} \Hyper{w}(z)$ is well-defined for general $t$) and $w\in T(\Sigma)$,
	\begin{equation}
		\AnaReg{z}{\infty} \Hyper{w}(z)
		\in L(\Sigma_t)(t) \tp \AnaReg{t}{0} \AnaReg{z}{\infty} L(\Sigma)(t)
		\label{eq:reginf-is-hyperlog}
	\end{equation}
	is itself a hyperlogarithm in $t$ with algebraic letters
	\begin{equation}
		\Sigma_t
		\defas
		\bigg\{
			\text{zeros of}\ 
			\prod_{i<j} \big[ \sigma_i(t) - \sigma_j(t) \big]
		\bigg\}
		\subset
		\overline{\Q}
		.
		\label{eq:def:reglim-letters}%
	\end{equation}%
\end{lemma}
Namely, equation~\eqref{eq:diff-reginf} presents an effective recursive algorithm to compute \eqref{eq:reginf-is-hyperlog}: First translate all shorter words
$
	\AnaReg{z}{\infty} \Hyper{\ldots\not\letter{\sigma_i}\ldots}(z) 
$
such that \eqref{eq:diff-reginf} reads
\begin{equation}
	\partial_t
	\AnaReg{z}{\infty} \Hyper{w} (z)
	=
	\sum_{u \in \Sigma_t^{\times}, \tau \in \Sigma_{t}}
		\frac{\lambda_{\tau, u}}{t - \tau}
		\Hyper{u} (t)
		\cdot c_{u}
\end{equation}
for some multiplicities $\lambda_{\tau,u} \in \Z$ and constants $c_u$. So
\begin{equation}
	\AnaReg{z}{\infty} \Hyper{w} (z)
	=
	C
	+
	\sum_{u \in \Sigma_t^{\times}, \tau \in \Sigma_t}
	\lambda_{\tau, u}
	\Hyper{ \letter{\tau} u} (t)
	\cdot c_{u}
\end{equation}
is determined up to the integration constant 
\begin{equation}
	C
	=
	\AnaReg{t}{0} \AnaReg{z}{\infty}
	\Hyper{w}(z)
	\quad\text{ by \eqref{eq:reg0}}.
	\label{eq:reglim-integration-constant} %
\end{equation}

\subsection{Regularized limits of regularized limits at infinity}
\label{sec:reglim-algorithm}%
We now explain how to compute the regularized limits \eqref{eq:reglim-integration-constant} symbolically, without a need for numeric evaluations (which are for example used in \cite{AnastasiouDuhrDulatMistlberger:SoftTripleRealHiggs} as explained in its appendix D). The ideas we present in the first half of this section were very recently also sketched in \cite{AblingerBluemleinRaabSchneiderWissbrock:Hyperlogarithms}.

Note that the limits $C$ of \eqref{eq:reglim-integration-constant} are constant only with respect to the variable $t$, while in our applications these will in general still depend on further variables.

Let $w \in \Sigma^{\times}$ be a word with letters $\Sigma \subset \Q(t) \setminus (0,\infty)$ depending rationally on $t$. 
We can restrict to $w = \WordReg{0}{}(w)$ not ending with $\letter{0}$ since $\AnaReg{z}{\infty} \Hyper{\letter{0}}(z) = 0$.
The simplest possible case is
\begin{lemma}
	\label{lemma:reglim-simple} %
	When $\lim_{t \rightarrow 0} (\Sigma) \subset \Q\setminus(0,\infty)$ is finite and $w = \letter{\sigma_1}\ldots\letter{\sigma_n}$ ends in a letter with $\lim_{t\rightarrow 0} (\sigma_n) \neq 0$, then
	\begin{equation}
		\AnaReg{t}{0} \AnaReg{z}{\infty} \Hyper{w} (z)
		=
		\AnaReg{z}{\infty}
		\Hyper{\lim\limits_{t \rightarrow 0} w}(z).
		\label{eq:reglim-simple} %
	\end{equation} %
\end{lemma}
\begin{proof}
	Using \eqref{eq:reginf} it suffices to investigate differences $w=(\letter{\sigma_1} - \letter{-1})\letter{\sigma_2} \ldots \letter{\sigma_n}$ and consider $\Hyper{w}(\infty)$ as the absolutely convergent integral \eqref{eq:reginf-convergent-integral} with integrand 
	\begin{equation}
		\label{eq:reginf-integrand-t-dependent} %
			\frac{1+\sigma_1(t)}{(z_1 + 1)(z_1 - \sigma_1(t))}
			\frac{1}{z_2-\sigma_2(t)}
			\cdots
			\frac{1}{z_n - \sigma_n(t)}.
	\end{equation}
	We can apply the theorem of dominated convergence to show $\lim_{t\rightarrow 0} \Hyper{w} (\infty) = \Hyper{\lim_{t\rightarrow 0} w} (\infty)$, essentially because the limiting integrand is absolutely integrable itself. So
	\begin{equation*}
		\AnaReg{t}{0} \AnaReg{z}{\infty} \Hyper{w} (z)
		=
		\AnaReg{t}{0} \Hyper{\WordReg{}{\infty}(w)} (\infty)
		=
		\Hyper{\lim\limits_{t \rightarrow 0} \WordReg{}{\infty}(w)}(\infty)
	\end{equation*}
	holds for all $w$ as allowed in \ref{lemma:reglim-simple}. Looking at \eqref{eq:def:word-reginf}, we may swap $\WordReg{}{\infty}$ and $\lim_{t\rightarrow 0}$ since the latter just substitutes letters $\letter{\sigma} \mapsto \letter{\sigma(0)}$. Finally apply \eqref{eq:reginf} again.
\end{proof}
This naive method can fail for three different reasons:
\begin{enumerate}[I]
	\item Some $\sigma(t) \in \Sigma$ diverges in the limit $t \rightarrow 0$.

	\item $\sigma_n(0) = 0$, because the limiting integrand \eqref{eq:reginf-integrand-t-dependent} at $t=0$ is not integrable:
		$\int_{0}^{z_{n-1}} \frac{\dd z_n}{z_n}$ diverges.

	\item Some letter has a limit $\sigma_i(0) \in (0,\infty)$ on the positive real axis, wherefore \eqref{eq:reginf-integrand-t-dependent} acquires a singularity at $z_i = \sigma_i(0)$ inside the domain of integration.
\end{enumerate}

We consider our main contribution as the algorithm to deal with these cases, which we present below. First note
\begin{lemma}[Scaling invariance]
	\label{lemma:scaling-invariance} %
	Given some $\alpha \in \Z$ and a word $w=\letter{\sigma_1}\ldots\letter{\sigma_n}$ in letters $\sigma_i \in \Q(t)$, let $w' \defas \letter{\sigma_1'}\ldots \letter{\sigma_n'}$ for $\sigma_i'(t) \defas t^{\alpha} \cdot \sigma_i(t)$. Then
	\begin{equation}
		\AnaReg{t}{0} \AnaReg{z}{\infty} \Hyper{w} (z)
		= \AnaReg{t}{0} \AnaReg{z}{\infty} \Hyper{w'} (z).
	\label{eq:scaling-invariance} %
	\end{equation} %
\end{lemma}
\begin{proof}
	Since $\AnaReg{z}{\infty} \Hyper{w}(z) = \AnaReg{z}{\infty} \Hyper{\WordReg{0}{} (w)}(z)$ we can restrict to $\sigma_n \neq 0$ and rescale $z' = z'' t^{-\alpha}$ in \eqref{eq:def:hyperlog} to conclude that $\Hyper{w}(z) = \Hyper{w'}\left( z t^{\alpha} \right)$.
	In regard of the regularization \eqref{eq:regularization} at $\tau=\infty$, this shows that
	\begin{equation*}
		f_{w,\infty}^{(i)} (z)
		=
		\sum_{j\geq i} \binom{j}{i} \log^{j-i} \left( t^{\alpha} \right)
		\cdot f_{w',\infty}^{(j)} \left(z t^{\alpha} \right).
	\end{equation*}
	Therefore we can conclude that indeed,
	\begin{align*}
		& \AnaReg{t}{0} f_{w, \infty}^{(0)}(\infty)
		= \AnaReg{t}{0} \sum_{i} 
				\left(\alpha \log t \right)^i 
				\cdot
				f_{w', \infty}^{(i)} (\infty)
		\\
		=& \AnaReg{t}{0} f_{w', \infty}^{(0)} (\infty)
		= \AnaReg{t}{0} \AnaReg{z}{\infty} \Hyper{w'} (z).
		\qedhere%
	\end{align*}
\end{proof}
A suitable such rescaling ensures finiteness of all $\sigma_i'(t)$ at $t \rightarrow 0$ and thus resolves problem I above.
\begin{example}
	\label{ex:scaling-invariance} %
	Applying lemmata \ref{lemma:scaling-invariance}, \ref{lemma:reglim-simple}, equation~\eqref{eq:reginf} and a M\"{o}bius transformation \ref{lemma:moebius-transform} with $f(z)=\frac{z}{1+z}$ shows
\begin{align*}
	&	\AnaReg{t}{0} \AnaReg{z}{\infty} \Hyper{\letter{-1}\letter{-1/t}}(z)
		=
		\AnaReg{t}{0} \AnaReg{z}{\infty} \Hyper{\letter{-t}\letter{-1}}(z)
		\\
	=&  \AnaReg{z}{\infty} \Hyper{\letter{0}\letter{-1}}(z)
		= \Hyper{(\letter{0}-\letter{-1})\letter{-1}}(\infty)
		= - \iInt_0^1 \!\letter{0}\letter{1}
		= \mzv{2}.
\end{align*} %
\end{example}
To address the issue II when $\sigma_n(0) = 0$, we make
\begin{definition}
	\label{def:reglim-laurent}%
	For any $0 \neq \sigma(t) \in \Q(t)$, the Laurent series
	$
		\sigma(t)
		=
		\sum_{n = N}^{\infty}
			t^{n} a_n
	$
	at $t\rightarrow 0$ with $a_N \neq 0$ defines a vanishing degree 
	$
		\deg_t (\sigma)
		\defas N
		\in \Z
	$
	and a leading coefficient 
	$
		\leadingCoefficient_t ( \sigma )
		\defas a_N 
		= \lim_{t\rightarrow 0} \left[ \sigma(t) \cdot t^{-N} \right]
	$. 
	For a word $w = \letter{\sigma_1}\ldots\letter{\sigma_n}$ set
	$
		\deg_t(w)
		\defas 
		\min \setexp{\deg_t(\sigma_i)}{1 \leq i \leq n}
	$.
\end{definition}
Whenever the final letter of a word $w=\letter{\sigma_1}\ldots\letter{\sigma_n}$ is of smallest vanishing degree  $\deg_t (\sigma_n) = \deg_t (w)$, rescaling $\sigma_k' \defas \sigma_k \cdot t^{-\deg_t(w)}$ ensures $\sigma_n'(0) = \leadingCoefficient_t(\sigma_n) \neq 0$ and lemma~\ref{lemma:reglim-simple} becomes applicable. Hence let us define
\begin{equation}
	\ReglimWord{t}{0} w
	\defas
	\lim_{t\rightarrow 0} \letter{\sigma_1'}\ldots\letter{\sigma_n'}
	\ \text{if}\ 
	\deg_t(w)
	=
	\deg_t(\letter{\sigma_n})
	,
	\label{eq:def:reglim-word} %
\end{equation}
e.g. $\ReglimWord{t}{0}(\letter{-1}\letter{-1/t}) = \letter{0}\letter{-1}$ in example~\ref{ex:scaling-invariance}.

But if $\deg_t(\sigma_n) > \deg_t(w)$, the rescaled $\sigma_n'(t)$ will vanish at $t \rightarrow 0$. In this case let
\begin{equation*}
	k
	\defas
	\max\setexp{i}{\deg_t (\sigma_i) = \deg_t (w)}
	<n
\end{equation*}
denote the last place in $w$ with minimal vanishing degree. Using \eqref{eq:shuffle-tail-decomposition} we can rewrite $w = \sum w_i \shuffle a_i$ such that each $w_i$ ends in $\letter{\sigma_k}$ and $a_i$ is a suffix of $\letter{\sigma_{k+1}}\ldots\letter{\sigma_n}$, i.e. $\abs{a_i}\leq n-k < n$. Applying this procedure recursively to each $a_i$ finally results in a representation 
\begin{equation}
	\label{eq:shuffle-decomposition-final-letter-minimal-vanishing} %
	w
	=
	\sum_i \left( w_{i,1} \shuffle \ldots \shuffle w_{i,r_i} \right)
\end{equation}
of $w$ in the shuffle algebra into elements $w_{i,j}$ each ending in some $\sigma_{i,j}$ with minimum vanishing degree $\deg_t\left( \sigma_{i,j} \right) = \deg_t\left( w_{i,j} \right)$.
\begin{example}
	For $w = \letter{-1}\letter{-t}$ this decomposition reads $w = \letter{-1} \shuffle \letter{-t} - \letter{-t}\letter{-1}$. So with $\AnaReg{z}{\infty} \Hyper{\letter{-1}}(z) = 0$,
	\begin{equation*}
		\AnaReg{t}{0} \AnaReg{z}{\infty} \Hyper{\letter{-1}\letter{-t}}(z)
		=
		-	\AnaReg{z}{\infty} \Hyper{\letter{0}\letter{-1}}(z).
	\end{equation*}
\end{example}
\begin{definition}
	\label{def:reglim-word} %
	For any alphabet $\Sigma \subset \Q(t)$ let
	\begin{equation}
		\label{eq:leading-alphabet} %
		\leadingCoefficient_t (\Sigma)
		\defas
		\set{0} \cupdot 
			\setexp{\leadingCoefficient_t(\sigma)}{\sigma \in \Sigma \setminus{0}}.
	\end{equation}
	Further we denote by $\ReglimWord{t}{0}: T(\Sigma) \longrightarrow T(\leadingCoefficient_t(\Sigma))$ the unique morphism of shuffle algebras that extends \eqref{eq:def:reglim-word} with $\ReglimWord{t}{0} (\letter{0}) = 0$.
\end{definition}
	This combinatorial regularized limit of words is a projection. For $w\in T(\Sigma)$ decomposed as \eqref{eq:shuffle-decomposition-final-letter-minimal-vanishing}, it is just
	\begin{equation}
		\ReglimWord{t}{0} w
		=
		\sum_i \left(
			\ReglimWord{t}{0} w_{i,1} \shuffle \ldots \shuffle \ReglimWord{t}{0} w_{i,r_i}
		\right).
		\label{eq:def:reglim-word-decomposed} %
	\end{equation}
	Putting together the lemmata \ref{lemma:scaling-invariance} and \ref{lemma:reglim-simple} with the linearity and multiplicativity of $\AnaReg{t}{0}$, $\AnaReg{z}{\infty}$ and $w \mapsto \Hyper{w}(z)$ we can compute regularized limits combinatorially using
\begin{corollary}
	\label{cor:reglim-word-no-positives} %
	For an alphabet $\Sigma \subset \C(t) \setminus (0,\infty)$ such that $\leadingCoefficient_t(\Sigma) \subset \C \setminus (0,\infty)$, any $w \in T(\Sigma)$ fulfills
	\begin{equation}
		\AnaReg{t}{0} \AnaReg{z}{\infty} \Hyper{w}(z)
		=
		\AnaReg{z}{\infty} \Hyper{\ReglimWord{t}{0}(w)}(z).
		\label{eq:reglim-word-no-positives} %
	\end{equation}
\end{corollary}

\begin{example}
	For $a,b,c \in \C\setminus[0,\infty)$, the decomposition of the form
	\eqref{eq:shuffle-decomposition-final-letter-minimal-vanishing} for	the word $	w=
	\letter{a}\letter{bt}\letter{ct^2} $ is
	\begin{equation*}
		w
		=
		\letter{a} \shuffle \letter{bt} \shuffle \letter{ct^2}
		 - \letter{a} \shuffle \letter{ct^2} \letter{bt}
		 - \letter{bt}\letter{a} \shuffle \letter{ct^2}
		 + \letter{ct^2}\letter{bt}\letter{a}
 \end{equation*}
	and shows
	$
		\AnaReg{t}{0}\AnaReg{z}{\infty} \Hyper{w}(z)
		=
		\AnaReg{z}{\infty} \Hyper{u}(z)
	$ 
	for
	\begin{equation*}
		u
		= \letter{a} \shuffle \letter{b} \shuffle \letter{c}
		 - \letter{a} \shuffle \letter{0} \letter{b}
		 - \letter{0}\letter{a} \shuffle \letter{c}
		 + \letter{0}\letter{0}\letter{a}.
	\end{equation*}
\end{example}
We still need to address problem III on our agenda: What happens when some $\leadingCoefficient_t(\sigma_{k}) \in (0,\infty)$ approaches a positive value?
In this situation, $\AnaReg{t}{0} \AnaReg{z}{\infty} \Hyper{w}(z)$ can have a discontinuity along $t\in[0,\infty)$. For example,
\begin{equation*}
	\AnaReg{z}{\infty} \Hyper{\letter{t}}(z)
	= \AnaReg{z}{\infty} \log \frac{z-t}{-t}
	= \log \frac{1}{-t}
	= \ln \abs{t} - \imag \arg (-t)
\end{equation*}
is defined only when $t \in \C \setminus [0,\infty)$; otherwise $\Hyper{\letter{t}}(z)$ is not well-defined for real $z\geq t$.
So to make sense of $\AnaReg{t}{0}$ we must tie $t \in \Halfplane^{\pm} \defas \setexp{z \in \C}{\pm \Imaginaerteil z > 0}$ to either the upper or lower half-plane resulting in
\begin{equation*}
	\AnaReg{t}{0+\imag \varepsilon} \AnaReg{z}{\infty} \Hyper{\letter{t}}(z)
	=	-\imag \pi
	\quad\text{and}\quad
	\AnaReg{t}{0-\imag \varepsilon} \AnaReg{z}{\infty} \Hyper{\letter{t}}(z)
	=	\imag \pi.
\end{equation*}
\begin{definition}
	\label{def:alphabet-above-below-partition} %
	Choosing $t \in \Halfplane^{+}$ partitions the alphabet
\begin{equation}
	\label{eq:alphabet-above-below-partition} %
	\Sigma
	=
	\widetilde{\Sigma} \cupdot \Sigma^{+} \cupdot \Sigma^{-}
	\subset
	\C(t) \setminus (0,\infty)
\end{equation}
into non-positive
$
	\leadingCoefficient_t ( \widetilde{\Sigma} )
	\cap
	(0,\infty)
	=
	\emptyset
$
and letters with 
$
	\leadingCoefficient_t ( \Sigma^{\pm} )
	\subset
	(0,\infty)
	$. These are separated by $\Sigma^{\pm} \subset \Halfplane^{\pm}$ for sufficiently small $\Realteil t$ and infinitesimal $\Imaginaerteil t > 0$.
	In particular we note that whenever $\leadingCoefficient_t(\sigma) \in (0,\infty)$,
	\begin{equation}
		\deg_t(\sigma)<0
		\Rightarrow
		\sigma \in \Sigma^{-}
		\ \text{and}\ 
		\deg_t(\sigma)>0
		\Rightarrow
		\sigma \in \Sigma^{+}
		.
	\end{equation}
	We denote the finite positive limits by
	\begin{equation}
		\label{eq:def:alphabet-finite-positive-limits} %
		\Sigma_0^{\pm}
		\defas
		\setexp{
				\lim_{t\rightarrow 0} \sigma(t)
		}{
				\sigma \in \Sigma^{\pm} 
				\ \text{and}\ 
				\deg_t(\sigma) = 0
		}.
	\end{equation}
\end{definition}
\begin{figure}
	\begin{center}
		\includegraphics[width=0.9\columnwidth]{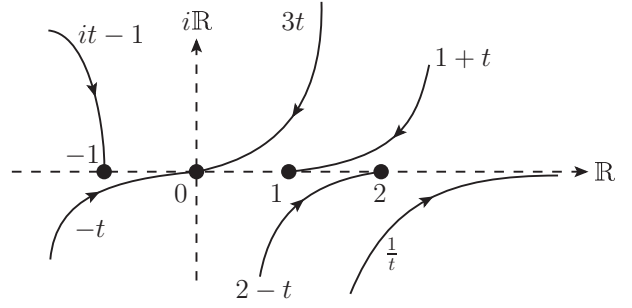}%
	\end{center}%
	\caption{This graph shows the limits of $\Sigma$ in example~\ref{ex:above-below-partition} when $t \rightarrow 0$ with positive real part and small positive imaginary part.}%
	\label{fig:above-below-limits}%
\end{figure}
\begin{example}
	\label{ex:above-below-partition} %
	For $\Sigma = \set{-1+\imag t,-t, 3t, 1+t,2-t, \frac{1}{t}}$, the limit $t\rightarrow  0 + \imag \varepsilon$ is shown in figure~\ref{fig:above-below-limits} and \eqref{eq:alphabet-above-below-partition} reads
	\begin{equation*}
		\widetilde{\Sigma}
		= \set{-1+\imag t, 0, -t}
		,
		\Sigma^{+}
		= \set{1+t, 3t}
		,
		\Sigma^{-}
		= \set{2-t,\frac{1}{t}}
		.
	\end{equation*} %
\end{example}
Now consider a word $w=\letter{\sigma_1}\ldots\letter{\sigma_n}$ with $\deg_t(\sigma_n) = \deg_t(w)$. Only the letters 
$
	\Sigma_w 
	\defas
	\setexp{\sigma_k}{\deg_t(\sigma_k) = \deg_t(w)}
$
play a role since after rescaling $\sigma_k'(t) = \sigma_k(t)\cdot t^{-\deg_t(w)}$ by lemma~\ref{lemma:scaling-invariance}, all other letters have $\deg_t(\sigma_k') = \deg_t(\sigma_k) - \deg_t(w) > 0$ and therefore approach $0 \notin (0,\infty)$ in the limit $t \rightarrow 0$ of \eqref{eq:reglim-simple}.
By homotopy invariance,
\begin{equation}
	\AnaReg{t}{0 + \imag\varepsilon}
	\AnaReg{z}{\infty} \Hyper{w}(z)
	=
	\iInt_{\gamma} \ReglimWord{t}{0} (w)
	\label{eq:reglim-contour-deformation} %
\end{equation}
is the iterated integral along a smooth deformation $\gamma$ of the originally real integration contour $[0,\infty)$. It avoids the positive letters among $\leadingCoefficient_t(\Sigma)$ as follows:
\begin{itemize}
	\item[]\hspace{-5mm}$\deg_t(w)<0$:
		$\Sigma_w \setminus \widetilde{\Sigma} \subset \Sigma^{-}$, $\gamma$ passes above all $\leadingCoefficient_t\left( \Sigma_w \right)$,

	\item[]\hspace{-5mm}$\deg_t(w)>0$:
		$\Sigma_w \setminus \widetilde{\Sigma} \subset \Sigma^{+}$, $\gamma$ passes below all $\leadingCoefficient_t\left( \Sigma_w \right)$,

	\item[]\hspace{-5mm}$\deg_t(w)=0$: 
		$	\lim_{t\rightarrow 0} (\Sigma_w \setminus \widetilde{\Sigma} )
			\subset
			\Sigma_0^{-} \cup \Sigma_0^{+}
		$ and $\gamma$ passes above $\Sigma_0^{-}$ and below $\Sigma_0^{+}$ as illustrated in figure~\ref{fig:contour-deformation} for example~\ref{ex:above-below-partition}. We must require that
		$
			\Sigma_0^{+} \cap \Sigma_0^{-}
			=
			\emptyset
		$
		as otherwise $\gamma$ is pinched between letters from $\Sigma^{+}$ and $\Sigma_{-}$ in the limit $t \rightarrow 0$. This situation did not occur in our applications but could be incorporated in the future.
\end{itemize}
\begin{figure}
	\begin{center}
		\includegraphics[width=0.9\columnwidth]{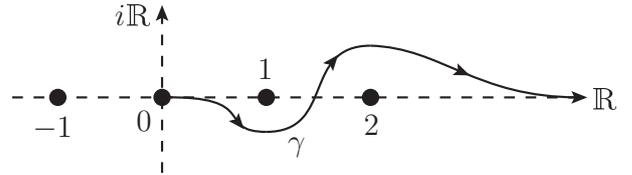}%
	\end{center}%
	\caption{The letters $\set{1+t,2-t} \subset \Sigma$ in example~\ref{ex:above-below-partition} induce a deformation of the real integration path $[0,\infty)$ towards $\gamma$, which avoids the positive limits in passing below $\Sigma_0^{+} = \set{1}$ and above $\Sigma_0^{-} = \set{2}$.}%
	\label{fig:contour-deformation}%
\end{figure}
In order to keep the implementation simple, we express \eqref{eq:reglim-contour-deformation} again in the form $\AnaReg{z}{\infty} \Hyper{v}(z)$ with $v$ not containing any positive letters (such that $\Hyper{v}(z)$ is single-valued on $z\in(0,\infty)$ and does not need additional specification of the contour $\gamma$). This is achieved by splitting up the contour $\gamma = \eta_u \concat \gamma_u$ at $u>0$ with the straight path $\eta_u$ from $0$ to $u$. So for $w=\letter{\sigma_1}\ldots\letter{\sigma_n}$ with $\sigma_n \neq 0$, Chen's lemma~\eqref{eq:Chens-lemma} takes the form
\begin{equation}
	\label{eq:chens-lemma-contour-splitting}%
	\iInt_{\gamma} w
	=
	\sum_{k=0}^{n}
		\iInt_{\gamma_u} \letter{\sigma_1}\ldots\letter{\sigma_k}
		\cdot
		\Hyper{\sigma_{k+1},\ldots,\sigma_n} (u)
\end{equation}
where $u < \tau \defas \min \left( \set{\sigma_1,\ldots,\sigma_n} \cap (0, \infty) \right)$ approaches the first potential branch point of $\Hyper{w}(z)$ on the positive real axis, see figure~\ref{fig:contour-split}. In the limit $u \rightarrow \tau$ we obtain
\begin{figure}
	\begin{center}
		\includegraphics[width=0.9\columnwidth]{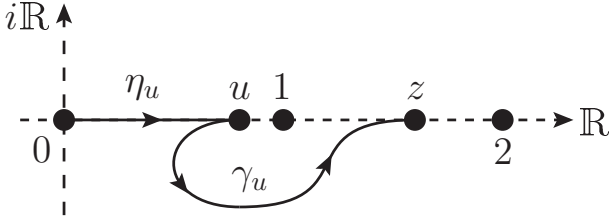}
	\end{center}
	\caption{The contour $\gamma$ of example~\ref{ex:above-below-partition} shown in figure~\ref{fig:contour-deformation} is homotopic to the splitting $\eta_u \concat \gamma_u$.}
	\label{fig:contour-split} %
\end{figure}
\begin{lemma}
	\label{lemma:contour-splitting} %
	Let $w=\letter{\sigma_1}\ldots\letter{\sigma_n} \in \Sigma^{\times}$ with $\sigma_n \neq 0$ and denote by
	$
		\tau
		\defas
		\min \left( 
			\Sigma
			\cap
			(0,\infty)
		\right)
	$
	the first (smallest) potential branch point of $\Hyper{w}(z)$ on the positive real axis.
	The analytic continuation past $\tau < z = \gamma(1)$ along $\gamma$ is
	\begin{equation}
		\Hyper{w}(z)
		=
		\sum_{k=0}^n
		\iInt_{\tau}^{z}
		(\letter{\sigma_1}\ldots\letter{\sigma_k})
		\cdot
		\Hyper{\WordReg{}{\tau} (\letter{\sigma_{k+1}}\ldots\letter{\sigma_n})} (\tau)
		\label{eq:contour-splitting} %
	\end{equation}
	where $\iInt_{\tau}^{z}$ denotes the iterated integral \eqref{eq:def:iterated-integral} whenever $\sigma_k \neq \tau$. It is extended to all words by imposition of $\iInt_{\tau}^z (w \shuffle v) = \iInt_{\tau}^{z} (w) \cdot \iInt_{\tau}^{z} (v)$ and letting $\gamma$ determine the branch of
	\begin{equation}
		\iInt_{\tau}^{z} \letter{\tau}
		\defas
		\iInt_{\gamma} \letter{\tau}
		=
		\pm \imag \pi
		+
		\log\frac{z-\tau}{\tau}
		\quad\text{for}\quad
		z > \tau.
		\label{eq:contour-splitting-branch} %
	\end{equation}
\end{lemma}
\begin{example}
	\label{ex:dilog-past-one} %
	Take the dilogarithm $\Hyper{\letter{0}\letter{1}}(z) = -\Li_2(z)$ and the path $\gamma$ passing below $\tau=1$ shown in figure~\ref{fig:contour-split}. Since $\WordReg{}{1} (\letter{1}) = 0$ and $\Hyper{\letter{0}\letter{1}}(1) = -\mzv{2}$, \eqref{eq:contour-splitting} reduces to
	\begin{equation*}
		\Li_2(z) - \mzv{2}
		=
		-\iInt_1^{z} (\letter{0}\letter{1})
		\defas
		-\iInt_1^{z} \letter{0} \cdot \iInt_1^{z} \letter{1} + \iInt_1^{z}(\letter{1}\letter{0}).
	\end{equation*}
	Inserting $\iInt_1^{z} \letter{1} \defas \iInt_{\gamma} \letter{1} = \imag\pi + \log(z-1)$and $z = 1+z'$, we obtain for $z'=z-1>0$ ($z>1$) the representation
	\begin{align*}
		\Li_2(z)
		&= \mzv{2} + \Hyper{\letter{0}\letter{-1}}(z') - (\imag\pi + \log z') \Hyper{\letter{-1}}(z')
		\\
		&= \mzv{2} - \Li_2(-z') - (\imag\pi + \log z')\log z.
	\end{align*}
\end{example}
In \eqref{eq:contour-splitting} we can resolve $\iInt_{\tau}^{z}$ into explicit factors $\pm\imag\pi$ from \eqref{eq:contour-splitting-branch} as dictated by $\gamma$ and the iterated integrals $\iInt_{\tau}^{z}$.

Recursive splitting of these at the next positive letter $\tau' \defas \min \left( \set{\sigma_1,\ldots,\sigma_n} \setminus {\tau} \cap (0,\infty) \right)$ finally expresses $\int_{\gamma} (w)$ in powers of $\pm\imag\pi$ and integrals $\iInt_{0}^{\tau}(v'), \iInt_{\tau}^{\tau'}(v''),\ldots$ which are simply defined by the straight line integration paths $0 \rightarrow \tau$, $\tau \rightarrow \tau'$ and so on. Through M\"{o}bius transformations \eqref{eq:moebius-transform}, these may all be transformed to $0\rightarrow \infty$.

\section{The implementation \HyperProg}
\label{sec:implementation}%
\subsection{General remarks}
We implemented the algorithms of section~\ref{sec:algorithms} in the computer algebra system {\Maple} \cite{Maple}. Even though these procedures are very flexible, we did not intend to provide a general purpose package supporting arbitrary symbolic calculations with hyper- and polylogarithms.

Instead, we were driven by our aim to compute Feynman integrals as we comment on in section~\ref{sec:feynman-integrals}. Therefore other applications are not as well supported, but we will give examples showing how {\HyperProg} can be used for quite general calculations with polylogarithms.

Note that we did not include facilities for numeric evaluations of hyper- and polylogarithms, because first of all these are not necessary for the algorithms and secondly there are already established programs available for this task, e.g. \cite{Ginac:Introduction,VollingaWeinzierl:NumericalMpl}.

The program uses the \code{remember} option of {\Maple}, which creates lookup tables to avoid recomputations of functions. But some of these functions depend on global parameters as explained for instance in section~\ref{sec:factorization}.
Therefore, whenever such a parameter is changed, the function \mbox{\code{forgetAll}$()$} must be called to invalidate those lookup tables. Otherwise the program might behave inconsistently.

\subsection{Installation and files}
\label{sec:installation}%
The program requires no installation. It is enough to load it during a {\Maple}-session by invoking
\begin{MapleInput}
read "HyperInt.mpl";
\end{MapleInput}
if the file \code{HyperInt.mpl} is located in the current directory or another place in the search paths of {\Maple}. If \code{periodLookups.m} can be found, it will be loaded automatically which is of great benefit as explained in section~\ref{sec:periods}.

All together, we supply the following main files:
\begin{description}
	\item[\Filename{HyperInt.mpl}]
		\hfill\\
		Contains our implementation of the algorithms in section~\ref{sec:algorithms} as well as supplementary procedures to handle Feynman graphs and Feynman integrals.

	\item[\Filename{periodLookups.m}]
		\hfill\\
		This table stores a reduction of multiple zeta values up to weight $12$ to a (conjectured) basis and similarly for alternating Euler sums up to weight $8$. It is not required to run the program, but necessary for efficient calculations involving high weights. Details follow in section~\ref{sec:periods}.

	\item[\Filename{Manual.mw}]
		\hfill\\
		This {\Maple} worksheet explains the practical usage of {\HyperProg}. In particular it includes plenty of explicit Feynman integral computations. Many explanations, details and comments are provided here.

	\item[\Filename{HyperTests.mpl}]
		\hfill\\
		A series of various tests of the program. Calling {\Maple} with \code{maple HyperTests.mpl} has to run without any error messages. Please report immediately if errors occur. Note that these tests only work when \code{periodLookups.m} can be found by {\HyperProg}.

		Due to the many different tests, including some Feynman integrals, the reader might find it instructive to read this file. See also section~\ref{sec:tests}.
\end{description}

\subsection{Representation of polylogarithms and conversions}
\label{sec:representations-conversions}
Internally, polylogarithms are represented as lists
\begin{equation}
	\label{eq:maple-internal-reginf-rep} %
	f
	=	
	[ 
		[g_1,[ w_{1,1}, \ldots, w_{1,r_1}] ],
		[g_2,[ w_{2,1}, \ldots, w_{2,r_1}] ],
		\ldots
	]
\end{equation}
of pairs of rational prefactors $g_i$ and lists of words $w_{i,j} = \WordReg{0}{}\left(w_{i,j}\right)$ not ending on $\letter{0}$. These encode the function
\begin{equation}
	\label{eq:maple-internal-reginf} %
	f
	=
	\sum_i g_i
				\cdot
				\prod_{j=1}^{r_i} \Hyper{\WordReg{}{\infty}(w_{i,j})}(\infty)
\end{equation}
and we decided not to use \eqref{eq:shuffle-product} to combine those words into the linear combination $\shuffle_j w_{i,j}$ for two reasons:
\begin{enumerate}
	\item Empirically, this expansion of shuffle products tends to increase the number of terms considerably.

	\item Our algorithm of section~\ref{sec:reglim-algorithm} to compute $\AnaReg{t}{0}$ produces products of words with \emph{different sets of letters}. Mixing these letters due to a shuffle introduces spurious letters in following integration steps which we want to avoid.
\end{enumerate}
To encode hyperlogarithms of a particular variable $z$, we use a list notation without explicit products:
\begin{equation}
	\label{eq:maple-internal-hyperlog} %
	f 
	= [ 
				[g_1, w_1],
				[g_2, w_2],
				\ldots
		]
	\defas
	\sum_{i} 
		g_i(z) \Hyper{w_i}(z)
	.
\end{equation}
These representations make the implementation of the algorithms of section~\ref{sec:algorithms} straightforward, but for easier, human-readable input and output we allow the notations
	\begin{align*}
		\code{Hlog}\left(z, [\sigma_1,\ldots,\sigma_r]\right)
		& \defas
		\Hyper{\sigma_1,\ldots,\sigma_r}(z)
		\quad\text{and}
		\\
		\code{Mpl}\left( [n_1,\ldots,n_r], [z_1,\ldots,z_r] \right)
		& \defas
		\Li_{n_1,\ldots,n_r}\left(z_1,\ldots,z_r\right)
	\end{align*}
for hyperlogarithms \eqref{eq:def:hyperlog} and multiple polylogarithms \eqref{eq:def:Li}. {\HyperProg} extends the native function \code{convert}$(f,\code{form})$ to transform an expression $f$ containing any of the functions
\begin{equation*}
\set{
	\code{log},
	\code{ln},
	\code{polylog},
	\code{dilog},
	\code{Hlog},
	\code{Mpl},
	\code{Hpl}
}
\end{equation*}
into one of the possible target formats
\begin{description}
	\item[$\code{form}=\code{HlogRegInf}$:]
		\hfill\\
		transforms $f$ into the list representation \eqref{eq:maple-internal-reginf-rep}.

	\item[$\code{form} \in \set{\code{Hlog}, \code{Mpl}}$:]
		\hfill\\
		expresses $f$ in terms of $\Hyper{}$ or $\Li$, using \eqref{eq:Hyper-as-Li}.

	\item[$\code{form} = \code{Hpl}$]
		\hfill\\
		translates hyperlogarithms $\code{Hlog}(z, w)$ with words $w \in \set{-1,0,1}^{\times}$ into the compressed notation of harmonic polylogarithms that was introduced in \cite{RemiddiVermaseren:HarmonicPolylogarithms}.
		Concretely, 
		$
			\code{Hpl}_{n_1,\ldots,n_r}(z)
			\defas
			\Hyper{\underline{n_1},\ldots,\underline{n_r}}(z)
		$
		where
		$
			\underline{0}
			\defas
			\letter{0}
		$
		and for any $n\in\N$,
		$
			\underline{\pm n}
			\defas
			\mp \letter{0}^{n-1} \letter{\pm 1}
	$.

	\item[$\code{form} = \code{i}$]
		\hfill\\
		same as $\code{form} = \code{Hlog}$, but produces the notation
		\begin{equation*}
			i[0,\sigma_n,\ldots,\sigma_1,z]
			\defas
			\code{Hlog}\left( z, [\sigma_1,\ldots,\sigma_n] \right)
		\end{equation*}
		which is used in {\zetaprocedures} \cite{Schnetz:ZetaProcedures}. The result can then be evaluated numerically in that program, e.g. using $\code{evalz}\left( \cdot \right)$.
\end{description}
\begin{example}The dilogarithm $\Li_2(z)$ has representations
\begin{MapleInput}
convert(polylog(2,z), Hlog);
\end{MapleInput}
\begin{MapleMath}
	-\Hlog{1}{0, 1/z}
\end{MapleMath}
\begin{MapleInput}
convert(polylog(2,z), HlogRegInf);
\end{MapleInput}
\begin{MapleMath}
	[[1, [[-1+z, -1]]], [-1, [[-1, -1]]]]
\end{MapleMath}
\end{example}
Due to the many functional relations, a general polylogarithm $f(\vec{z})$ has many different representations. In particular, the representation \eqref{eq:maple-internal-reginf} is far from being unique.

It is therefore crucial to be able to express polylogarithms in a basis in order to simplify results and to detect relations.
As was demonstrated in \cite{Brown:TwoPoint}, lemma~\ref{lemma:reginf-is-hyperlog} provides such a basis through
\begin{corollary}
	\label{corollary:fibration-basis} %
	Let $f(\vec{z}) = \AnaReg{z}{\infty} \Hyper{w}(z)$ for $w \in L(\Sigma)$ with rational letters $\Sigma \subset \C(\vec{z})$ and choose an order $\vec{z}=(z_1,\ldots,z_n)$. Then there is a unique way to write
	\begin{equation}
		f(\vec{z}) 
		=
		\sum_i
			\Hyper{w_{i,1}}(z_1)
			\cdot
			\ldots
			\cdot
			\Hyper{w_{i,n}}(z_n)
			\cdot
			c_i
		\label{eq:fibration-basis} %
	\end{equation}
	as a linear combination of products of hyperlogarithms of words $w_{i,j} \in T(\Sigma_i)$ with letters in some algebraic alphabets
	$
		\Sigma_i
		\subset
		\overline{\C(z_{i+1},\ldots,z_n)}
	$,
	which may only depend on the following variables. The factors $c_i$ in \eqref{eq:fibration-basis} are constants (with respect to $\vec{z}$), namely
	\begin{equation}
		\label{eq:fibration-basis-constants} %
		c_i
		\in
		\AnaReg{z_n}{0} \ldots \AnaReg{z_1}{0}
		\AnaReg{z}{\infty} L(\Sigma)(z).
	\end{equation}
\end{corollary}
Its implementation constitutes the essential function
\begin{equation*}
	\code{fibrationBasis}\left(f, [z_1,\ldots,z_r], F\right),
\end{equation*}
which writes a polylogarithm $f$ (preferably in the list notation \eqref{eq:maple-internal-reginf-rep}, otherwise it will be converted first) in the form \eqref{eq:fibration-basis} with respect to the order $\vec{z} = [z_1,\ldots,z_r]$ of variables (when $\vec{z}$ is omitted, $\vec{z}=[]$ is used).
If the optional table $F$ is supplied, the result will be stored as $F_{[w_{i,1},\ldots,w_{i,n}]} = c_i$. 
		
\begin{example}
	This function can be used to obtain functional relations between polylogarithms. For example,
	\begin{MapleInput}
fibrationBasis(polylog(2,1-z), [z]);
convert(%, Mpl);
	\end{MapleInput}
	\begin{MapleMath}
		-\Hlog{z}{1, 0} + \mzv{2}
		\\
		-\Mpl{2}{z} + \ln(z) \Mpl{1}{z} + \mzv{2}
	\end{MapleMath}%
	reproduces the classic identity $\Li_2(1-z) = \mzv{2} - \Li_2(z) - \log z \log(1-z)$. Similarly, we obtain the inversion relation for $\Li_5\left(-\frac{1}{x} \right) = \frac{1}{120} \ln^5 x + \frac{\mzv{2}}{6} \ln^3 x + \frac{7}{10}\mzv[2]{2} \ln x + \Li_5(-x)$:
	\begin{MapleInput}
fibrationBasis(polylog(5, -1/x), [x]):
convert(%, Mpl);
	\end{MapleInput}
	\begin{MapleMath}
		\frac{1}{6}\mzv{2}\ln(x)^3+\frac{1}{120} \ln(x)^5 + \Mpl{5}{-x} + \frac{7}{10} \mzv[2]{2} \ln(x)
	\end{MapleMath}
	As an example involving multiple variables, the five-term relation of the dilogarithm is recovered as
\begin{MapleInput}
polylog(2,x*y/(1-x)/(1-y))-polylog(2,x/(1-y))-polylog(2,y/(1-x)):
fibrationBasis(%, [x, y]);
\end{MapleInput}
	\begin{MapleMath}
		\Hlog{y}{0, 1}+\Hlog{x}{0, 1}-\Hlog{x}{1}\Hlog{y}{1}
	\end{MapleMath}
\end{example}
Note that for more than one variable, each choice $\vec{z}$ of order defines a different basis and a function may take a much simpler form in one basis than in another. 
For example, $\Li_{1,2}(y,x) + \Li_{1,2}(\frac{1}{y}, xy)$ is just
\begin{MapleInput}
f:=Mpl([1,2], [y,x])+Mpl([1,2], [1/y,y*x]):
fibrationBasis(f, [x,y]);
\end{MapleInput}
\begin{MapleMath}
	\Hlog{x}{0, 1/y, 1}+\Hlog{x}{0, 1, 1/y}
\end{MapleMath}
but in another basis takes the form
\begin{MapleInput}
fibrationBasis(f, [y,x]);
\end{MapleInput}
\begin{MapleMath}
	\Hlog{y}{0, 1, 1/x}+\Hlog{y}{0, 1/x}\Hlog{x}{1}
	\\
	-\Hlog{y}{0, 0, 1/x}-\Hlog{y}{0, 1}\Hlog{x}{1}
\end{MapleMath}
We like to emphasize that every order $\vec{z}$ defines a true basis without relations. In particular this means that $f=0$ if and only if $\code{fibrationBasis}(f, \vec{z})$ returns $0$, no matter which order $\vec{z}$ was chosen.

Analytic continuation in a variable $z$ is performed along a straight path, therefore the result can be ambiguous when this line contains a point where the function is not analytic. In this case, an auxiliary variable
\begin{equation}
	\delta_z
	=
	\begin{cases}
		+1 & \text{when $z \in \Halfplane^{+}$,} \\
		-1 & \text{when $z \in \Halfplane^{-}$} \\
	\end{cases}
	\label{eq:which-half-plane}
\end{equation}
will appear to distinguish the branches above and below the real axis. From example \ref{ex:dilog-past-one} consider
\begin{MapleInput}
fibrationBasis(polylog(2, 1+z), [z]);
\end{MapleInput}
\begin{MapleMath}
	I\pi\delta_{z} \Hlog{z}{-1}-\Hlog{z}{-1, 0} + \mzv{2}
\end{MapleMath}

\subsection{Periods}
\label{sec:periods}%
Our algorithms express constants like \eqref{eq:fibration-basis-constants} through iterated integrals $\AnaReg{0}{\infty} \Hyper{w}(z)$ of words $w \in \overline{\Q}^{\times}$ with algebraic letters. These are transformed into iterated integrals $\Hyper{u}(1)$ by $u=\code{zeroInfPeriod}(w)$.
Such special values of multiple polylogarithms satisfy a huge number of relations and it is clearly highly desirable to express them in a basis over $\Q$.

The case $u\in\set{0,1}^{\times}$ of multiple zeta values (MZV) is by now perfectly understood on the motivic level \cite{Brown:MixedTateMotivesOverZ}, such that conjectural $\Q$-bases are available at arbitrary weight and \cite{Brown:DecompositionMotivicMZV} even provides a reduction algorithm that was implemented in \cite{Schnetz:ZetaProcedures}. Similar results can also be found for some cases of $u\in\setexp{0,\mu}{\mu^N=1}^{\times}$ with $N$-th roots of unity $\mu$, see \cite{Deligne:GroupeFondamentalMotiviqueN}.

{\HyperProg} can load lookup tables to benefit from such relations and we supply the file \Filename{periodLookups.m} which provides the reductions that were proven in the data mine project \cite{BluemleinBroadhurstVermaseren:Datamine} using standard relations. It includes multiple zeta values up to weight 12 and alternating Euler sums ($u\in\set{-1,0,1}^{\times}$) up to weight $8$ in the notation
\begin{equation}
	\mzv{n_1,\ldots,n_r}
	\defas
	\Li_{\abs{n_1},\ldots,\abs{n_r}}\left( \frac{n_1}{\abs{n_1}},\ldots \frac{n_r}{\abs{n_r}} \right),
	\label{eq:def:euler-sum} %
\end{equation}
with indices $n_1,\ldots,n_r \in \Z \setminus\set{0}$, $n_r\neq 1$.
When $u \in\set{0,a,2a}^{\times} \cup \set{-a,0,a}^{\times}$, M\"{o}bius transformations are used to express $\Hyper{u}(1)$ in terms of alternating Euler sums and $\log a$.

\begin{example}
	{\HyperProg} automatically attempts to load \Filename{periodLookups.m}, but can run without it. With its help,
	\begin{MapleInput}
fibrationBasis(Mpl([3], [1/2]));
	\end{MapleInput}
	\begin{MapleMath}
		\frac{1}{6}\ln(2)^3-\frac{1}{2}\ln(2)\mzv{2}+\frac{7}{8}\mzv{3}
	\end{MapleMath}
	is reduced to MZV and $\ln 2$. But if \Filename{periodLookups.m} is not available, we obtain merely
	\begin{MapleInput}
fibrationBasis(Mpl([3], [1/2]));
	\end{MapleInput}
	\begin{MapleMath}
		-\mzv{-3}-\mzv{2,-1}-\mzv{1,-2} + \frac{1}{6}\ln(2)^3
	\end{MapleMath}
\end{example}

The user can define a different basis reduction or provide bases for periods involving higher weights\footnote{For MZV and alternating sums, \cite{BluemleinBroadhurstVermaseren:Datamine} provides reductions up to weights 22 and 12, respectively.}, or additional letters. These must be defined as a table,
\begin{equation}
	\code{zeroOnePeriods}[u]
	\defas
	\Hyper{u}(1),
\end{equation}
and saved to a file $f$. To read it call $\code{loadPeriods}(f)$.

\begin{example}
	Polylogarithms $\Li_{\vec{n}}(\vec{z})$ at fourth roots of unity $\vec{z} \in \set{\pm 1, \pm \imag}^{\abs{n}}$ up to weight $\abs{n}\leq 2$, like
	\begin{MapleInput}
f := Mpl([1,1],[I,-1])+Mpl([1,1],[-1,I]):
fibrationBasis(f);
	\end{MapleInput}
	\begin{MapleMath}
		\Hlog{1}{-I, I}+\Hlog{1}{-1, I}
	\end{MapleMath}
	are tabulated in \Filename{periodLookups4thRoots.mpl} in terms of $\ln 2$, $\imag$, $\pi$ and Catalan's constant $\Imaginaerteil\Li_2(\imag)$:
	\begin{MapleInput}
loadPeriods("periodLookups4thRoots.mpl"):
fibrationBasis(f);
	\end{MapleInput}
	\begin{MapleMath}
		\frac{1}{8}\mzv{2}+\frac{1}{2}\ln(2)^2-\frac{1}{4}I\pi\ln(2)+I\Catalan
	\end{MapleMath}
\end{example}

\subsection{Integration of hyperlogarithms}
\label{sec:integration}%
The most important function provided by {\HyperProg} is
\begin{equation}
	\label{eq:integrationStep}%
	\code{integrationStep}(f, z)
	\defas
	\int_0^{\infty} f(z)\ \dd z
\end{equation}
and computes the integral of a polylogarithm $f$, which must be supplied in the form \eqref{eq:maple-internal-reginf-rep}.
First it explicitly rewrites $f(z) \in L(\Sigma)(z)$ following lemma~\ref{lemma:reginf-is-hyperlog} as a hyperlogarithm in $z$. Then a primitive $F=\code{integrate}(f,z)$ is constructed as explained in section~\ref{sec:integration-differentiation} and finally expanded at the boundaries $z\rightarrow 0, \infty$.
\begin{example}
	To compute $\int_0^{\infty} \frac{\Li_{1,1}\left(-x/y, -y \right)}{y(1+y)} \dd y$, type
	\begin{MapleInput}
convert(Mpl([1,1],[-x/y,-y])/y/(y+1), HlogRegInf): integrationStep(%, y):
fibrationBasis(%, [x]);
	\end{MapleInput}
	\begin{MapleMath}
		\mzv{2}\Hlog{x}{1}
		 + \Hlog{x}{1,0,1}
		 - \Hlog{x}{0,0,1}
	\end{MapleMath}
\end{example}
A more convenient and flexible form is the function
\begin{equation}\begin{split}
	&\code{hyperInt}\left( f, [z_1=a_1..b_1,\ldots,z_r=a_r..b_r] \right)
	\\
	&\qquad
	\defas
	\int_{a_r}^{b_r}
	\cdots
	\left[\int_{a_1}^{b_1} f\ \dd z_1 \right]
	\cdots
	\dd z_r
	\label{eq:hyperInt}%
\end{split}\end{equation}
which computes multi-dimensional integrals by repeated application of \eqref{eq:integrationStep} in the order $z_1,\ldots,z_r$ as specified. It automatically transforms the domains $(a_k,b_k)$ of integration to $(0,\infty)$ and furthermore, $f$ can be given in any form that is understood by $\code{convert}\left( \cdot, \code{HlogRegInf} \right)$.
\begin{example}
	\label{ex:moduli-space}%
	A typical integral studied in the origin \cite{Brown:MZVPeriodsModuliSpaces} of the algorithm is $I_2$ of equation~(8.6) therein:
	\begin{MapleInput}
I2 := 1/(1-t1)/(t3-t1)/t2:
hyperInt(I2, [t1=0..t2, t2=0..t3, t3=0..1]):
fibrationBasis(%);
	\end{MapleInput}
	\begin{MapleMath}
		2\mzv{3}
	\end{MapleMath}
\end{example}
\begin{example}
The ``Ising-class'' integrals $E_n$ were defined in \cite{BaileyBorweinCrandall:IsingClass}: For $n\geq 2$ let $u_k \defas \prod_{i=2}^k t_i$, $u_1 \defas 1$ and set
\begin{equation}
	E_n
	\defas
	2 \int_0^{1} \dd t_2 \ldots \int_0^1 \dd t_n
	\left( 
		\prod_{1 \leq j< k \leq n}
		\frac{u_j - u_k}{u_j+u_k}
	\right)^2.
	\label{eq:def:IsingE}%
\end{equation}
Because the denominators $u_j+u_k = (1+\prod_{i={j}}^{k-1} t_i) \prod_{i=k}^{n} t_i$ have very simple factors, it is easy to prove linear reducibility along the sequence $t_2,\ldots,t_n$ and to show that all $E_n$ are rational linear combinations of alternating Euler sums.

We included a simple procedure \mbox{\code{IsingE}$\left( n \right)$} to evaluate them in the attached manual. In particular we can confirm the conjecture on $E_5$ made in \cite{BaileyBorweinCrandall:IsingClass}:
\begin{MapleInput}
IsingE(5);
\end{MapleInput}
\begin{MapleMath}
2\mzv{3}
\left(
 - 37
 + 232\ln (2) 
\right)
 - 4\mzv{2}
\left(
31
 - 20\ln  ( 2 ) 
 + 64 \ln^2  ( 2 )
\right)
\\
 - \frac{318}{5}\mzv[2]{2}
 + 42
 - 992\mzv{1,-3}
 - 40 \ln ( 2 )
 + 464 \ln^2 ( 2 )
 + \frac{512}{3} \ln^4 ( 2 )
\end{MapleMath}
For illustration further exact results for $E_n$ up to $n=8$ can be found in \Filename{IsingE.mpl}. Time- and memory-requirements of these computations are summarized in table~\ref{tab:IsingE-resources}.
\end{example}
\begin{table*}
	\centering
	\begin{tabular}{rcccccccc}
		\toprule
		$n$ & 1 & 2 & 3 & 4 & 5 & 6 & 7 & 8 \\
		\midrule
		time & \SI{10}{\milli\second} & \SI{41}{\milli\second} & \SI{52}{\milli\second} & \SI{235}{\milli\second} & \SI{2.0}{\second} & \SI{40.6}{\second} & \SI{29.3}{\minute} & \SI{28}{\hour}\\
		RAM & \SI{35}{\mebi\byte} & \SI{51}{\mebi\byte} & \SI{51}{\mebi\byte} & \SI{76}{\mebi\byte} & \SI{359}{\mebi\byte} & \SI{1.6}{\gibi\byte} & \SI{1.9}{\gibi\byte} & \SI{30}{\gibi\byte}	\\
		\bottomrule
	\end{tabular}
	\caption{Resources consumed during computation of the Ising-type integrals $E_n$ of \eqref{eq:def:IsingE} running on Intel\textsuperscript{\textregistered} Core{\texttrademark} i7-3770 CPU @ \SI{3.40}{\giga\hertz}. The column with $n=1$ (when $E_n \defas 1$) requires no actual computation and shows the time and memory needed to load \Filename{periodLookups.m}.}%
	\label{tab:IsingE-resources}%
\end{table*}

\subsubsection{Singularities in the domain of integration}
\label{sec:integration-contour-deformation}%
The integration \eqref{eq:integrationStep} requires that $f(z) \in L(\Sigma)(z)$ is a hyperlogarithm without any letters $\Sigma_+ \defas \Sigma \cap (0,\infty) = \emptyset$ inside the domain of integration, which ensures that $f(z)$ is analytic on $(0,\infty)$.

Otherwise $f(z)$ can have poles or branch points on $\Sigma_+$ and the integration is then performed along a deformed contour $\gamma$ as discussed in section~\ref{sec:reglim-algorithm}. The dependence on $\gamma$ (see figure~\ref{fig:contour-split}) is encoded in the variables
\begin{equation}
	\label{eq:contour-deformation-deltas}%
	\delta_{z,\sigma}
	=
	\begin{cases}
		+1 & \text{when $\gamma$ passes below $\sigma$},\\
		-1 & \text{when $\gamma$ passes above $\sigma$}.\\
	\end{cases}
\end{equation}
\begin{example}
	\label{ex:contour-deformation-integral}%
	The integrand $f(z)=\frac{1}{1-z^2}$ has a simple pole at $z\rightarrow 1$ and is not integrable over $(0,\infty)$. Instead, {\HyperProg} computes the contour integrals
	\begin{MapleInput}
hyperInt(1/(1-z^2), z): fibrationBasis(%);
	\end{MapleInput}
	\begin{MapleOutput}
Warning, Contour was deformed to avoid potential singularities at {1}.
	\end{MapleOutput}
	\begin{MapleMath}
		-\frac{1}{2} \cdot I \pi \delta_{z, 1}
	\end{MapleMath}
\end{example}
Note even when positive letters $\Sigma_{+}$ occur, $f(z)$ can be analytic on $(0,\infty)$ nonetheless. In this case the dependence on any $\delta_{z,\sigma}$ drops out in the result.
\begin{example}
	The integrand $f(z)=\frac{\ln(z)}{1-z^2}$ is analytic at $z\rightarrow 1$ and thus on all of $(0,\infty)$. It integrates to
	\begin{MapleInput}
hyperInt(ln(z)/(1-z^2), z):
fibrationBasis(%);
	\end{MapleInput}
	\begin{MapleOutput}
Warning, Contour was deformed to avoid potential singularities at {1}.
	\end{MapleOutput}
	\begin{MapleMath}
		-\frac{3}{2} \mzv{2}
	\end{MapleMath}
\end{example}

\subsubsection{Detection of divergences}
\label{sec:divergence-detection}%
By default, the option $\code{\_hyper\_check\_divergences}=\code{true}$ is activated and triggers, after each integration, a test of convergence. The primitive $F(z)$ is expanded as
\begin{equation}
	\label{eq:check-divergences}%
	F(z)
	= \sum_{i=0}^{N} \log^i z \sum_{j=-M}^{\infty} z^{j} F_{i,j}
	\quad\text{at}\quad
	z \rightarrow 0
\end{equation}
and all polylogarithms $F_{i,j}$ with $i>0$ or $j<0$ are explicitly checked to vanish $F{i,j}=0$ using \code{fibrationBasis}; the limit $z \rightarrow \infty$ is treated analogously. This method is time-consuming and we recommend to deactivate this option for any involved calculations, expecting that the convergence is granted by the problem at hand.
\begin{example}
	\label{ex:divergence-check}%
	An endpoint divergence at $z\rightarrow \infty$ is detected for $\int_0^{\infty} \frac{\ln z}{1+z}\dd z = \lim_{z\rightarrow\infty} \Hyper{\letter{-1}\letter{0}}(z)$:
	\begin{MapleInput}
hyperInt(ln(z)/(1+z), z);
	\end{MapleInput}
	\begin{MapleError}
Error, (in integrationStep) Divergence at z = infinity of type ln(z)^2
	\end{MapleError}
\end{example}
The expansions \eqref{eq:check-divergences} are only performed up to $i,j\leq\code{\_hyper\_max\_pole\_order}$ (default value is $10$). If higher order expansions are needed, an error is reported and this variable must be increased.

Note that the expansion \eqref{eq:check-divergences} is only computed at the endpoints $z\rightarrow 0,\infty$. Polar singularities inside $(0,\infty)$ are not detected, e.g.\ $\code{hyperInt}\left(\frac{1}{(1-z)^2}, z\right) = \restrict{\frac{1}{1-z}}{0}^{\infty} = 1$ calculates the integral along a contour evading $z=1$ just as discussed in section~\ref{sec:integration-contour-deformation}.
One can split the integration
\begin{equation}
	\label{eq:integral-splitted}%
	\int_0^{\infty} f(z) \ \dd z
	=
	\sum_{i=0}^{k} \int_{\tau_i}^{\tau_{i+1}} f(z)\ \dd z
\end{equation}
at such critical points $\Sigma_+ = \set{\tau_1<\ldots<\tau_k}$ with $\tau_0 \defas 0$, $\tau_{k+1} \defas \infty$ with the effect that all singularities now lie at endpoints and will be properly analyzed by the program.

A problem arises if calculations involve periods for which no basis reduction is known to {\HyperProg}, because the vanishing $F_{i,j}=0$ of a potential divergence might not be detected. One can then set $\code{\_hyper\_abort\_on\_divergence} \defas \code{false}$ to continue with the integration. All $F_{i,j}$ of \eqref{eq:check-divergences} are stored in the table \code{\_hyper\_divergences}.
\begin{example}
	\label{eq:divergence-period-relations}%
	When \Filename{periodLookups.m} is not loaded,
	\begin{MapleInput}
hyperInt(polylog(2,-1/z)*polylog(2,-z)/z,z);
	\end{MapleInput}
	\begin{MapleError}
Error, (in integrationStep) Divergence at z = infinity of type ln(z)
	\end{MapleError}
	inadvertently finds a divergence. Namely, $F_{1,0}$ of \eqref{eq:check-divergences} is
	\begin{MapleInput}
entries(_hyper_divergences, pairs);
	\end{MapleInput}
	\begin{MapleMath}
		\left(z=\infty, \ln \left( z \right) \right)
		=
		4\mzv{1,3} + 2 \mzv{2,2} - \frac{1}{36} \pi^{4}
	\end{MapleMath}
	and its vanishing corresponds to an identity of MZV.
\end{example}
We like to remark that through this observation, the computation of an integral which is known to be finite in fact implies some relations among periods.

\subsection{Factorization of polynomials}
\label{sec:factorization}%
Since we are working with hyperlogarithms throughout, it is crucial that all polynomials occurring in the calculation factor linearly with respect to the integration variable $z$. For example,
\begin{MapleInput}
integrationStep([[1/(1+z^2), []]], z);
\end{MapleInput}
\begin{MapleError}
Error, (in partialFractions) 1+z^2 is not linear in z
\end{MapleError}
fails because factorization is initially only attempted over the rationals $\K=\Q$. Instead we can allow for an algebraic extension $\K=\Q(R)$ by specification of a set $R=\code{\_hyper\_splitting\_field}$ of radicals:
\begin{MapleInput}
_hyper_splitting_field := {I}:
integrationStep([[1/(1+z^2), []]], z);
fibrationBasis(%);
\end{MapleInput}
\begin{MapleMath}
	\left[
		\left[
			\frac{1}{2}I, [[-I]]
		\right],
		\left[
			-\frac{1}{2}I, [[I]]
		\right]
	\right]
	\\
	\frac{1}{2} \pi
\end{MapleMath}
We can also go further and factorize over the full algebraic closure $\K=\overline{\Q(\vec{z})}$ by setting $\code{\_hyper\_algebraic\_roots}\defas\code{true}$. Over $\K$, all rational functions $\Q(\vec{z})$ factor linearly such that we can integrate any $f \in \AnaReg{t}{\infty} L(\Sigma)(t)$ as long as we start with rational letters $\Sigma\subset \Q(\vec{z})$.

This feature is to be considered experimental and only applied in \code{transformWord} which implements lemma~\ref{lemma:reginf-is-hyperlog}: Given an irreducible polynomial $P \in \Q[\vec{z}]$ and a distinguished variable $z$, the symbolic notation
\begin{equation}
	\label{eq:root-letter}%
	\letter{\code{Root}(P,z)}
	\defas
	\sum
	\setexp{\letter{z_0}}{\restrict{P}{z=z_0} = 0}
\end{equation}
sums the letters corresponding to all the roots of $P$.
\begin{example}
	\label{ex:algebraic-letters}%
	A typical situation looks like this:
	\begin{MapleInput}
f,g:=Hlog(x,[-z,x+x^2]),Hlog(x,[x+x^2,-z]):
fibrationBasis(f+g, [x, z]);
	\end{MapleInput}
	\begin{MapleError}
Error, (in linearFactors) z+x+x^2 does not factor linearly in x
	\end{MapleError}
	To express $f+g$ as a hyperlogarithm in $x$, the roots $R=\code{Root}(P,x)=\set{-\frac{1\pm\sqrt{1-4z}}{2}}$ of $P=z+x+x^2$ seem necessary. After allowing for such algebraic letters, we obtain:
	\begin{MapleInput}
_hyper_algebraic_roots := true:
fibrationBasis(f+g, [x, z]);
	\end{MapleInput}
	\begin{MapleMath}
		- \Hlog{x}{-1,-z}
		 - \Hlog{x}{-z,-1}
		 \\
		 + \Hlog{x}{-z,0}
		 + \Hlog{x}{0,-z}
	\end{MapleMath}
	Since this result actually does not involve $\letter{R}$ at all one might wonder why it was necessary in the first place. The reason is that the individual contributions $f$ and $g$ indeed need $\letter{R}$. Only in their sum this letter drops out:\footnote{%
	In this extremely simple example this is clear since by \eqref{eq:Chens-lemma},
	$
		f+g
		=
		\Hyper{\letter{-z}}(x)
		\cdot
		\Hyper{\letter{x(x+1)}}(x)
	$
	factorizes into $\log \frac{x+z}{z} \cdot \log \frac{x}{1+x}$. We thus see why our representation \eqref{eq:maple-internal-reginf-rep} is preferable to one where all products of words are multiplied out (as shuffles).%
}%
	\begin{MapleInput}
alias(R = Root(z+x+x^2, x)):
fibrationBasis(f, [x, z]);
	\end{MapleInput}
	\begin{MapleMath}
		\Hlog{x}{R,-z}
		 + \Hlog{x}{R,-1}
		 - \Hlog{x}{R,0}
		\\
		 + \Hlog{x}{-z,0}
		 - \Hlog{x}{-z,-1}
		\\
		 - \Hlog{x}{-1,-z}
		 - \Hlog{x}{0,-z}
	\end{MapleMath}
\end{example}
Note that further processing of functions with such algebraic letters\footnote{%
These are sometimes referred to as generalized harmonic polylogarithms with \emph{nonlinear weights}.%
} is not supported by {\HyperProg}, because their integrals are in general not hyperlogarithms anymore. However, the case of example~\ref{ex:algebraic-letters} occurs frequently wherefore the option \code{\_hyper\_ignore\_nonlinear\_polynomials} (default value is \code{false}) is available to ignore all algebraic letters in the first place. That is, all words containing such a letter are immediately dropped when it is set to \code{true}.

In the example above this gives the correct result for $f+g$, but will provoke false answers when \code{fibrationBasis} is applied to $f$ or $g$ alone. Hence this option should only be used when linear reducibility is granted; preferably using the methods of section~\ref{sec:polynomial-reduction}.

\subsection{Additional functions}

In the manual we describe some further procedures provided by {\HyperProg} (note that all algorithms of section~\ref{sec:algorithms} are were implemented), like the extension of the commands \code{diff} and \code{series} to compute differentials and series expansions of hyperlogarithms.

\subsection{Performance}
\label{sec:performance}%
During programming we focussed on correctness and we are aware of considerable room for improvement of the efficiency of {\HyperProg}.
But we hope that our code and the details provided in section~\ref{sec:algorithms} will inspire further, streamlined implementations, even outside the regime of computer algebra systems.
This is possible since apart from the factorization of polynomials (which can be performed before the actual integration, see the next section), all operations boil down to elementary manipulations of words (lists) and computations with rational functions.

Ironically, often just decomposing into partial fractions becomes a severe bottleneck in practice, as was also noted in \cite{AblingerBluemleinRaabSchneiderWissbrock:Hyperlogarithms}. This happens when an integrand contains denominator factors to high powers or very large polynomials in the numerator.

We observed that {\Maple} consumes a lot of main memory, in very challenging calculations the demand grew beyond $\SI{100}{\gibi\byte}$. Often this turns out to be the main limitation in practice.

Our program uses some functions that are not thread-safe and can therefore not be parallelized automatically. However, since the integration procedure considers every hyperlogarithm individually, a manual parallelization is straightforward: Multiple instances of {\Maple} can each compute a different piece of an integral whose results can be added up afterwards. Some example scripts are provided and discussed in the manual.

Also note that the product representation \eqref{eq:maple-internal-reginf-rep} inherently allows for different representations of the same words, because a product can either be represented symbolically or as the corresponding sum of shuffles. 
We argued that shuffling out every product is not desirable, so a better solution could be to choose an order on the alphabet $\Sigma$, which then gives rise to a polynomial basis of the shuffle algebra $T(\Sigma)$ in terms of Lyndon words \cite{Radford:BasisShuffleAlgebra}.

\section{Polynomial reduction and linear reducibility}
\label{sec:polynomial-reduction}%
In order to compute multi-dimensional integrals \eqref{eq:iteration-partial-integrals} by iterated integration using the algorithms of section~\ref{sec:algorithms}, we must require that for each $k$, the partial integral
\begin{equation}
	f_k
	\in
	L\left( \Sigma_k \right) (z_{k+1})
	\ \text{where}\ 
	\Sigma_k
	\subset 
	\C\left( z_{k+2}, \ldots, z_n \right)
	\label{eq:partial-integrals-hyperlogs} %
\end{equation}
is a hyperlogarithm in the next integration variable $z_{k+1}$. The alphabet $\Sigma_k$ is restricted to rational functions of the remaining variables, in particular $\Sigma_k \subset \C(z_{k+2})$, because only then lemma~\ref{lemma:reginf-is-hyperlog} guarantees that its integral $f_{k+1} \in L\left( (\Sigma_k)_{z+2} \right)(z_{k+2})$ is a hyperlogarithm in $z_{k+2}$.
\begin{definition}
	We call $f_0(\vec{z})$ \emph{linearly reducible} if for some ordering $z_1,\ldots,z_n$ of its variables, sets $\Sigma_k$ exist such that \eqref{eq:partial-integrals-hyperlogs} holds for all $0\leq k < n$.
\end{definition}
For illustration let us suppose we want to integrate
\begin{equation}
	f_0(x,y,z)
	\defas
	\frac{1}{((1+x)^2+y)(y+z^2)}
	\label{eq:reduction-example-0}%
\end{equation}
over $x$ and $y$. To integrate $x$, we must include in $\Sigma_x$ the algebraic zeros $-1\pm\imag\sqrt{y}$ to get a hyperlogarithm $f_0(x) \in L\left( \Sigma_x \right) (x)$ in $x$. But then the integral
\begin{equation}
	\int_0^{\infty} f_0\ \dd x
	=
	\frac{\arctan\sqrt{y}}{\sqrt{y}(y+z^2)}
	\label{eq:reduction-example-x}%
\end{equation}
is not a hyperlogarithm in $y$ at all\footnote{But it is a hyperlogarithm in $t\defas\sqrt{y}$, so in this simple case a change of variables would help us out.}. On the other hand, since $f_0(y) \in L\left( \set{-(1-x)^2,-z^2} \right)(y)$ for letters rational in $x$, integration of $y$ results in a hyperlogarithm
\begin{equation}
	\int_0^{\infty} f_0 \ \dd y
	=
	2\frac{\log(1+x) -\log z}{(x+1+z)(x+1-z)}
	\label{eq:reduction-example-y}%
\end{equation}
in $x$ over letters $\set{-1,-1\pm z}$. So in the order $z_1 \defas y$, $z_2 \defas x$ linear reducibility is given and we can integrate
\begin{equation}
	\int_0^\infty \dd x
	\int_0^{\infty} \dd y
	\ f_0
	=
	\frac{\Hyper{\letter{1}\letter{0}}(z)-\Hyper{\letter{-1}\letter{0}}(z)}{z},
	\label{eq:reduction-example-xy}%
\end{equation}
which is a harmonic polylogarithm in $z$.

In principle, we can try to integrate $f_0$ for some arbitrary order and verify, after each step, that $\Sigma_k$ is rational (or otherwise abort and try a different order). But fortunately this is not necessary since there are means to analyse the singularities of the integrals $f_k$ in advance.

Namely, \emph{polynomial reduction} algorithms were presented in \cite{Brown:TwoPoint} and \cite{Brown:PeriodsFeynmanIntegrals}. These compute, for each subset $I\subset E \defas \set{z_1,\ldots,z_n}$ of variables, a set 
$
	S_I
	\subset
	\Q[E\setminus I]
$
of irreducible polynomials that provide an upper bound of the Landau varieties as introduced in \cite{Brown:PeriodsFeynmanIntegrals}. In particular this means that if there exists an ordering $z_1,\ldots,z_n$ of the variables such that all $p \in S_{I_k}$ are linear in $z_{k+1}$, for any $0
\leq k < n$ and $I_k \defas \set{z_1,\ldots,z_{k}}$, then the linear reducibility \eqref{eq:partial-integrals-hyperlogs} is granted with the rational alphabets
\begin{equation}
	\Sigma_k
	\defas
	\set{0}
	\cup
	\bigcup_{p \in S_{I_k}}
	\set{\text{zeros of $p$ in w.r.t.\ $z_{k+1}$}}.
	\label{eq:polynomial-reduction-alphabets}%
\end{equation}
Readers familiar with the symbol calculus will realize that the polynomials $S_{I_k}$ provide an upper bound of the entries of the symbol of $f_k$.
Explicit examples of such reductions are worked out in \cite{Brown:TwoPoint,Brown:PeriodsFeynmanIntegrals} and the appendix of \cite{Panzer:DivergencesManyScales}.

\subsection{Performance}
A polynomial reduction can significantly speed up computations of integrals \eqref{eq:iteration-partial-integrals}: During the step when $f_k$ is rewritten as a hyperlogarithm in $z_{k+1}$ following section~\ref{sec:reginf-as-hyperlog}, all words that contain a letter not in $\Sigma_k$ can be dropped, since the knowledge of \eqref{eq:polynomial-reduction-alphabets} proves that all such contributions must in total add up to zero (see example~\ref{ex:algebraic-letters}, where the algebraic roots $\letter{\code{Root}(P,z)}$ drop out for $f+g$).

Note that the dimension of the space of hyperlogarithms over an alphabet $\Sigma_k$ grows exponentially with the weight. Therefore, a polynomial reduction is absolutely crucial for problems of high complexity and cutting down the number of polynomials in $\Sigma_k$ is highly desirable. In practice this means that after computation of a polynomial reduction, one should look for a sequence $z_1,\ldots,z_n$ of variables not only ensuring that $S_{I_k}$ are linear in $z_{k+1}$, but also minimizing the number of $z_{k+1}$-dependent polynomials in $S_{I_k}$.

\subsection{Implementation in {\HyperProg}}
{\HyperProg} implements the \emph{compatibility graph} method \cite{Brown:PeriodsFeynmanIntegrals} of polynomial reduction and provides it as the command $\code{cgReduction}\left( L \right)$. The entries $L_I=\left[S_I, C_I\right]$ of the table $L$ are pairs of polynomials $S_I$ and edges $C_I \subset \binom{S_I}{2}$ between them.
\begin{example}
	\label{ex:polynomial-reduction}%
	The reduction of the integrand \eqref{eq:reduction-example-0} starts with the complete graph on the factors of its denominator:
	\begin{MapleInput}
S:={x^2+2*x+1+y,y+z^2}: L[{}]:=[S, {S}]:
cgReduction(L):
L[{x}][1]; L[{y}][1]; L[{x,y}][1];
	\end{MapleInput}
	\begin{MapleMath}
		{L_{\set{x}}}_1
		\\
		\set{x+1,x+1+z,x+1-z}
		\\
		\set{1+z,z-1}
	\end{MapleMath}
	We see that the results for $S_{\set{y}}$ and $S_{\set{x,y}}$ match with the letters of \eqref{eq:reduction-example-y} and \eqref{eq:reduction-example-xy}, but $S_{\set{x}}$ is not computed because $S_{\emptyset}$ is not linear in $x$.
\end{example}
Our implementation can use the knowledge of such reductions in two places (examples are given in the manual):
\begin{itemize}
	\item
		When a table $S$ is supplied as the (optional) fourth parameter to \code{fibrationBasis}, then all words $w_{i,k}$ in \eqref{eq:fibration-basis} containing letters not in $\Sigma_k$ of \eqref{eq:polynomial-reduction-alphabets} are removed from the result.

	\item
		In the first step of integrating $\int_0^{\infty} f\ \dd z$, the integrand $f$ is rewritten as a hyperlogarithm in $z$ using $\code{transformWord}(f, z) = \sum_w \Hyper{w}(z) \cdot c_u$. Setting
\begin{MapleInput}
_hyper_restrict_singularities := true:
_hyper_allowed_singularities := S:
\end{MapleInput}
ensures that any word $w$ containing a letter that is not a zero of some polynomial $p(z) \in S$ is dropped.
\end{itemize}

\subsection{Spurious polynomials and changes of variables}
Bear in mind that the sets $S_I$ only provide upper bounds on the alphabet. In course of our calculations we regularly observed that, with the number $\abs{I}$ of integrated variables increasing, more and more polynomials in $S_I$ tend to be spurious. In extreme cases it happens that a reduction contains surplus non-linear polynomials in every variable, while $f_0$ actually is linearly reducible.

But even when linear reducibility strictly fails, it is sometimes possible to change variables such that the integrand becomes linearly reducible in these new variables. 
We explain this in \cite{Panzer:DivergencesManyScales} using the example of a divergent, massive four-point box integral. Similar transformations are also employed in \cite{AblingerBluemleinRaabSchneiderWissbrock:Hyperlogarithms} to calculate generating functions of operator insertions into finite one-scale integrals. Also note the discussion \cite{BoncianiDegrassiVicini:GeneralizedHarmonicPolylogarithms} of alphabets containing square root letters that are typical for applications in particle physics and can be rationalized through simple changes of variables.

\section{Application to Feynman integrals}
\label{sec:feynman-integrals}%
In section~\ref{sec:algorithms} we investigated hyperlogarithms on their own, but the algorithms were originally developed in \cite{Brown:TwoPoint} for the computation of Feynman integrals. Important results on their linear reducibility (including counterexamples) and the geometry of Feynman graph hypersurfaces were obtained in \cite{Brown:PeriodsFeynmanIntegrals}. 
In \cite{Panzer:MasslessPropagators,Panzer:DivergencesManyScales} we successfully applied our implementation to compute many non-trivial examples, including massless propagators up to six loops and also divergent integrals depending on up to seven kinematic invariants. All results\footnote{These can be downloaded from \url{http://www.math.hu-berlin.de/~panzer/}.} presented in these papers were computed using this prgram {\HyperProg}.

Some further discussions on multi-scale and subdivergent integrals in the parametric representation are also given in \cite{BognerLueders:MasslessOnShell,BrownKreimer:AnglesScales,Kreimer:WheelsInWheels}.

We hope that our implementation will be particularly useful for applications to particle physics.

\begin{remark}
Our method applies only to the small class of linearly reducible graphs, which is a subset of those Feynman graphs that can be evaluated in terms of polylogarithms.
By now it is however well known that quantum field theory exceeds this space of functions not only in the massive case \cite{LaportaRemiddi:AnalyticSunrise,AdamsBognerWeinzierl:Sunrise,BlochVanhove:Sunset}, but also in massless integrals \cite{BrownDoryn:FramingsForGraphHypersurfaces}. Even in supersymmetric theories, elliptic integrals and generalizations have been identified, e.g.\ \cite{NandanPaulosSpradlinVolovich:StarIntegrals,CaronHuotLarsen:UniquenessTwoLoopMasterContours}.
\end{remark}

\subsection{Parametric representation and $\varepsilon$-expansion}
The popular method of Schwinger parameters \cite{ItzyksonZuber} expresses Feynman integrals $\Phi(G)$ associated to Feynman graphs $G$ by
\begin{equation}
	\Phi(G)
	= \Gamma(\sdd)
		\prod_{e \in E} \int_0^{\infty} \frac{\SP_e^{\EP_e-1}\ \dd \SP_e}{\Gamma(\EP_e)}
		\cdot \frac{\phipol^{-\sdd}}{\psipol^{\Dim/2-\sdd}}
		\cdot \delta(1-\SP_{e_N})
	\label{eq:parametric-feynman-integral}
\end{equation}
in $\Dim$ space-time dimensions. To each edge $e \in E$ of the graph corresponds a Schwinger variable $\SP_e$, and the corresponding scalar propagator may be raised to some power $\EP_e$. 
The superficial degree of divergence is 
$
	\sdd 
	\defas
	\sum_{e\in E} \EP_e
	- \loops{G} \cdot \frac{\Dim}{2}
$
for the loop number $\loops{G}$ of $G$.
The two graph polynomials $\psipol$ and $\phipol$ are for example defined in \cite{BognerWeinzierl:GraphPolynomials}, the $\delta$-distribution freezes an arbitrary $\SP_{e_N}$.

\subsection{$\varepsilon$-expansion}
For calculations in dimensional regularization\footnote{A definition in momentum space can be found in \cite{Collins}, while in the parametric representation it is immediate.}, we set $\Dim=4-2\varepsilon$ and also the edge powers $\EP_e = \EPZ_e + \varepsilon\EPE_e$ are $\varepsilon$-dependent and expanded near an integer $\EPZ_e \in \Z$. Assuming that \eqref{eq:parametric-feynman-integral} is convergent\footnote{This can always be arranged for with the help of preparatory partial integrations as was shown in \cite{Panzer:DivergencesManyScales}.} for $\varepsilon=0$, we can expand the integrand in $\varepsilon$ and obtain each coefficient $c_n$ of the Laurent series $\Phi(G) = \sum_{n} c_n \varepsilon^n$ as period integrals
\begin{equation}
	c_n
	= \Gamma(\sdd)
	\prod_{e\in E}
	\int_0^{\infty}
	\frac{\dd \SP_e}{\Gamma(\EP_e)}
	\cdot
	\frac{P^{(n)} \cdot f^{(n)}}{Q^{(n)}}
	\delta(1-\SP_{e_N})
	\label{eq:expansion-coefficients}%
\end{equation}
where $P^{(n)},Q^{(n)} \in \Q[\vec{\SP}]$ denote polynomials and $f^{(n)} \in \Q[\vec{\SP}, \log\vec{\SP},\log\phipol,\log\psipol]$. In particular $f^{(n)} \in L(\Sigma_e)(\SP_e)$ is a hyperlogarithm in $\SP_e$ whenever $\phipol$ and $\psipol$ are linear in $\SP_e$. If $f^{(n)}$ even turns out to be linearly reducible, we can integrate it with {\HyperProg}.

\subsection{Additional functions in {\HyperProg}}
In \ref{sec:function-list-feynman} we list the most important functions that support the calculation of Feynman integrals. These entail simple routines to construct the graph polynomials $\psipol$ and $\phipol$.

For divergent integrals, the parametric integrands in the representation \eqref{eq:parametric-feynman-integral} can be divergent. Such a situation demands partial integrations, which effectively implement the analytic (dimensional) regularization and produce a convergent integral representation in the end. This procedure is defined and exemplified in \cite{Panzer:DivergencesManyScales} and implemented into {\HyperProg} as described in the manual.

\begin{figure}
	\centering
		\includegraphics[width=0.8\columnwidth]{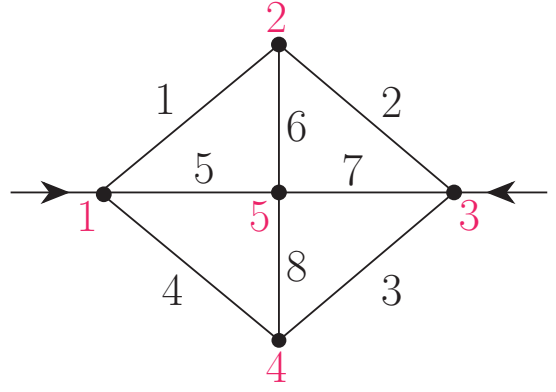}
		\caption{Four-loop massless propagator of section~\ref{sec:example-propagator}. In \cite{Panzer:MasslessPropagators} this one is called $M_{3,6}$. Edges are labelled in black, vertices in red.}%
	\label{fig:propagator-graph}%
\end{figure}
\subsection{Examples}
\label{sec:example-propagator}%
Plenty of examples are provided in the {\Maple} worksheet \Filename{Manual.mw}, wherefore we only present a very brief case of a four-loop massless propagator here.

First we define the graph of figure~\ref{fig:propagator-graph} by its edges $E$ and specify two external momenta of magnitude one entering the graph at the vertices $1$ and $3$. The polynomials $\psipol$ and $\phipol$ can be computed with
\begin{MapleInput}
E:=[[1,2],[2,3],[3,4],[4,1],[5,1],[5,2], [5,3],[5,4]]:
psi:=graphPolynomial(E):
phi:=secondPolynomial(E, [[1,1], [3,1]]):
\end{MapleInput}
This graph has vertex-width three \cite{Brown:PeriodsFeynmanIntegrals} and is therefore linearly reducible. Still let us calculate a polynomial reduction to verify this claim:
\begin{MapleInput}
L:=table(): S:=irreducibles({phi,psi}):
L[{}]:=[S, {S}]: cgReduction(L):
\end{MapleInput}
Afterwards we can investigate the polynomial reduction (for example with the procedure \mbox{$\code{reductionInfo}(L)$}) and find a linearly reducible sequence $\vec{z}$ of variables. We recommend to always check this with
\begin{MapleInput}
z:=[x[1],x[2],x[6],x[5],x[3],x[4],x[7],x[8]]:
checkIntegrationOrder(L, z[1..7]):
\end{MapleInput}
\begin{MapleOutput}
1. (x[1]): 2 polynomials, 2 dependent
2. (x[2]): 5 polynomials, 4 dependent
3. (x[6]): 8 polynomials, 4 dependent
4. (x[5]): 7 polynomials, 4 dependent
5. (x[3]): 6 polynomials, 6 dependent
6. (x[4]): 4 polynomials, 3 dependent
7. (x[7]): 1 polynomials, 1 dependent
Final polynomials:
\end{MapleOutput}
\begin{MapleMath}
	\set{}
\end{MapleMath}
The integrand is assembled according to \eqref{eq:parametric-feynman-integral} which in this case is already convergent as-is. We expand to second order in $\varepsilon$ with
\begin{MapleInput}
sdd := nops(E)-(1/2)*4*(4-2*epsilon):
f := series(psi^(-2+epsilon+sdd)*phi^(-sdd), epsilon=0):
f:=add(coeff(f,epsilon,n)*epsilon^n,n=0..2):
\end{MapleInput}
Now we integrate out all but the last Schwinger parameter
\begin{MapleInput}
hyperInt(f, z[1..-2]):
\end{MapleInput}
and reduce the result into a basis of MZV:
\begin{MapleInput}
fibrationBasis(f)*z[-1]:
collect(%, epsilon);
\end{MapleInput}
\begin{MapleMath}
	\left(
		254\mzv{7}
		+780 \mzv{5}
		-200 \mzv{2}\mzv{5}
		-196 \mzv[2]{3}
		+80 \mzv[3]{2}
		-\frac{168}{5}\mzv[2]{2}\mzv{3}
	\right)
	\varepsilon^2
	\\
+ \left(
		-28 \mzv[2]{3}
		+140 \mzv{5}
		+\frac{80}{7} \mzv[3]{2}
	\right)
	\varepsilon
+20 \mzv{5}.
\end{MapleMath}
Examples containing more external momenta, massive propagators and also divergences are included in \Filename{Manual.mw}.

\paragraph{Acknowledgments}
I thank Francis Brown for his beautiful articles and Dirk Kreimer for continuous encouragement.
Oliver Schnetz kept me interested into graphical functions and kindly verified many of my computations with his very own methods, thereby providing a strong cross-check.
Also Johannes Henn provided some $\varepsilon$-expansions to me for tests, and he motivated the study of divergent integrals in the parametric representation.
Many discussions with Christian Bogner of concrete examples and problems greatly improved my understanding of iterated integrals.
Figures were generated with {\JaxoDraw} \cite{BinosiTheussl:JaxoDraw}.

\appendix

\section{Tests of the implementation}
\label{sec:tests}%
We extensively tested our implementation with a variety of examples. Most of these are supplied in the file \Filename{HyperTests.mpl} which must run without any errors. Since it contains many diverse applications of {\HyperProg}, it might also be useful as a supplement to the manual.

Plenty of functional and integral equations of polylogarithms, taken from the books \cite{Lewin:PolylogarithmsAssociatedFunctions,Lewin:StructuralPropertiesPolylogarithms}, are checked with {\HyperProg}. These tests revealed a few misprints in \cite{Lewin:PolylogarithmsAssociatedFunctions}:
\begin{itemize}
	\item
	Equation (7.93): $-\frac{9}{4}\pi^2 \log^2(\xi)$ must be $-\frac{9}{12}\pi^2 \log^2(\xi)$.

	\item
	Equation (7.99), repeated as (44) in appendix~A.2.7: The second term $-\frac{9}{4}\pi^2 \log^3(\xi)$ of the last line must be replaced with $-\frac{3}{4}\pi^2 \log^3(\xi)$.

	\item
	Equation A.3.5. (9): The terms $-2 \Li_3\left( 1/x \right) + 2\Li_3(1)$ should read $+\Li_3(1/x) - \Li_3(1)$ instead.

	\item
	In equation (7.132), a factor $\frac{1}{2}$ in front of the second summand $D^n_{p=0} \frac{1}{p} \left\{ \cdots \right\}$ is missing (it is correctly given in 7.131).

	\item
	Equation (8.80): $(1-v)$ inside the argument of the fourth $\Li_2$-summand must be replaced by $(1+v)$, so that after including the corrections mentioned in the following paragraph, the correct identity reads
	\begin{equation}\begin{split}
		0=&
		\Li_2\left( \frac{(1+v)w}{1+w} \right)
		+ \Li_2\left( \frac{-(1-v)w}{1-w} \right)
		\\
		+& \Li_2\left( \frac{(1-v)w}{1+w} \right)
		+ \Li_2\left( \frac{-(1+v)w}{1-w} \right)
		\\
		- &\Li_2\left( \frac{-(1-v^2)w^2}{1-w^2} \right)
		+ \frac{1}{2} \log^2\left( \frac{1+w}{1-w} \right).
		\label{eq:lewin-8.80-corrected}%
	\end{split}\end{equation}

	\item
	Equation (16.46) of \cite{Lewin:StructuralPropertiesPolylogarithms}: $x^2$ must read $x^{-2}$.

	\item
	Equation (16.57) of \cite{Lewin:StructuralPropertiesPolylogarithms}: $\frac{\pi^4}{40}$ must read $\frac{\pi^4}{30}$.
\end{itemize}
Some tests are constructed by calculation of parametric integrals with known results in terms of polylogarithms and MZV. We used the expansion of Euler's beta function in the form
\begin{equation*}
	\frac{\exp\left[ \sum\limits_{n=2}^{\infty} \frac{\zeta(n)}{n}( x^n + y^n-(x+y)^n)  \right]}{1-x-y}
	=
	\int_0^{\infty} \frac{z^{-x}\ \dd z}{(1+z)^{2-x-y}}
\end{equation*}
and also checked the identity ($z \geq 0$)
\begin{align}
	&
	\int_{0}^{\infty} \left[ 
		\left( \frac{1}{x} - \frac{1}{x+z} \right) \Li_n(-x-z)
		-\frac{1}{x} \Li_n\left( -\frac{z}{x+1} \right)
	\right] \dd x
	\nonumber\\
	& =
	n \Li_{n+1}(-z),
	\label{eq:bubble-chain-identity} %
\end{align}
which is easily derived inductively for any $n\in \N$.
\begin{figure*}%
	\begin{equation*}
		\BBr{n}{m}
		\defas
		\Graph[0.5]{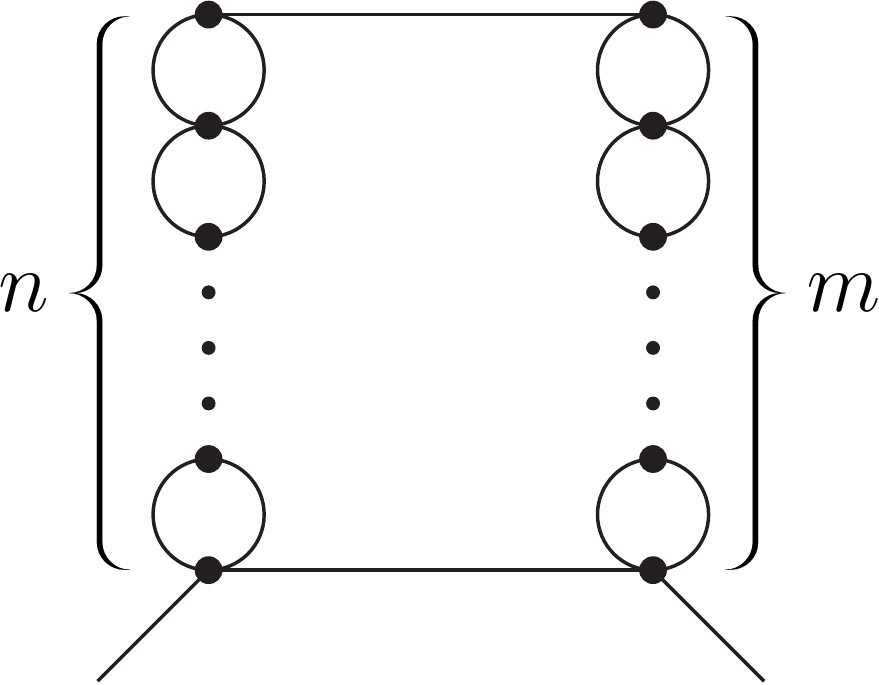}
		\qquad
		\BBt{n}{m}
		\defas
		\Graph[0.5]{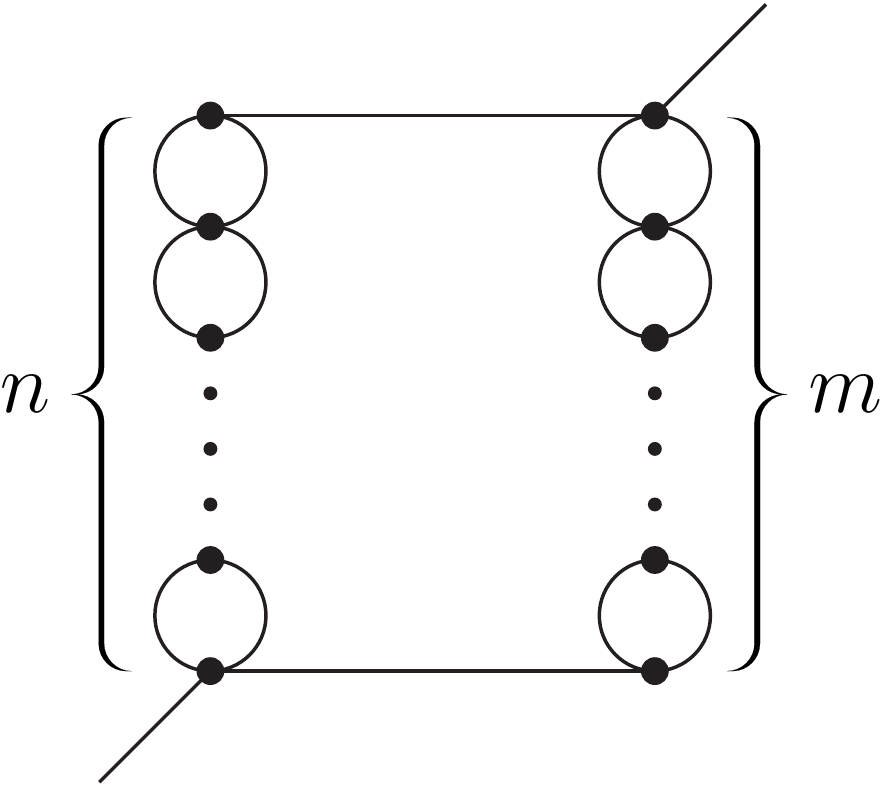}
	\end{equation*}%
	\caption{Two series of one-scale graphs with subdivergences in four dimensions. They occur in $\phi^4$-theory as vertex graphs with two nullified external momenta, incident to the two three-valent vertices.}%
	\label{fig:bubble-chains}%
\end{figure*}

The two families of ``bubble chain graphs'' shown in figure~\ref{fig:bubble-chains} can be calculated with standard techniques in momentum-space. Following the forest formula, we get
\begin{equation}
	\period\left( \BBr{n}{m} \right)
	\defas
	\restrict{\frac{\partial}{\partial q^2}}{q^2=1}
	\Phi_R\left( \BBr{n}{m} \right)
	= (n+m)!
	\label{eq:bubble-chains-rational-period}%
\end{equation}
for the derivative of the Feynman integrals $\Phi_R$ renormalized by subtraction at external momentum $q^2 = 1$. The second family has generating function
\begin{multline}
	\ln \bigg\{ 
		(1-x-y)
		\sum_{n,m\geq 0} \frac{x^n y^m}{n! m!} \period\left( \BBt{n}{m} \right)
	\bigg\}
	\\
	=
	\sum_{r\geq 1} \frac{2\mzv{2r+1}}{2r+1} \left[ 
		x^{2r+1} + y^{2r+1} - (x+y)^{2r+1}
	\right]
	\label{eq:bubble-chains-transcendental-period}%
\end{multline}
and we used \eqref{eq:bubble-chains-rational-period} and \eqref{eq:bubble-chains-transcendental-period} to verify our results obtained from the parametric integral representations for these periods derived in \cite{Panzer:Mellin}.

We furthermore tested some simple period integrals of \cite{Brown:MZVPeriodsModuliSpaces} and transformations of polylogarithms into hyperlogarithms given in \cite{BroedelSchlottererStieberger:PolylogsMZVSuperstringAmplitudes}. Our results for the integrals $E_n$ of \eqref{eq:def:IsingE} match the analytic results up to $n=4$ given in \cite{BaileyBorweinCrandall:IsingClass} and the numeric values obtained therein for $E_5$ and $E_6$ agree with our exact results.

Probably the strongest tests of our implementation are the computations of $\varepsilon$-expansions of various single-scale \cite{Panzer:MasslessPropagators} and multi-scale \cite{Panzer:DivergencesManyScales} Feynman integrals. We cross-checked these results with many different references, verified that they obey the symmetries of the associated Feynman graphs and in some cases used established programs \cite{BognerWeinzierl:ResolutionOfSingularities,SmirnovTentyukov:FIESTA} to obtain numeric evaluations to confirm our analytic formulas.

Also we confirmed the operator matrix elements $\hat{I}_{1a}$, $\hat{I}_{1b}$, $\hat{I}_{2a}$, $\hat{I}_{2b}$, $\hat{I}_{4}$ of ladder graphs computed in \cite{Wissbrock:Massive3loopLadder} and checked the Benz graphs $I_1$, $I_2$ and $I_3$ of \cite{AblingerBluemleinRaabSchneiderWissbrock:Hyperlogarithms}. The examples of $\hat{I}_4$ and $I_1$ are part of our manual, where we also correct mistakes in the equations (3.18) and (3.1) loc.cit.

Another check was done with the massless hexagon integral \cite{DelDucaDuhrSmirnov:OneLoopMasslessHexagon,DixonDrummondHenn:OneLoopMasslessHexagon}, which is also included in the manual.

\section{Proofs}
\label{sec:proofs} %

\begin{proof}[Lemma \ref{lemma:shuffle-head-tail-decomposition}]
	The statement is trivial for $n=0$ and we apply induction. For $n>0$, the outer shuffle product in the right-hand side of \eqref{eq:shuffle-tail-decomposition} decomposes with respect to the last letter into
	\begin{multline*}
		\left\{
		\sum_{i=0}^{n-1}
		\left[ u \shuffle ( -\letter{a_i}) \ldots (-\letter{a_1} ) \right]
		\letter{\sigma}
		\shuffle
		\letter{a_{i+1}}\ldots\letter{a_{n-1}}
		\right\}
		\letter{a_n}
		\\
		+
		\left\{
			u
			\shuffle
			\sum_{i=0}^{n}
				(-\letter{a_i}) \ldots (-\letter{a_1})
				\shuffle
				\letter{a_{i+1}}\ldots\letter{a_{n}}
		\right\}
		\letter{\sigma}.
	\end{multline*}
	The first contribution is $u\letter{\sigma}\letter{a_1}\ldots\letter{a_n}$ by the induction hypothesis and the second contribution vanishes because it represents 
	$
		\left\{ 
				u
				\shuffle 
				(S\convolution \id) (\letter{a_1}\ldots\letter{a_{n}})
		\right\}
		\letter{\sigma}
	$
	for the antipode $S$ of the Hopf algebra $T(\Sigma)$.
\end{proof}

\begin{proof}[Lemma \ref{lemma:word-reginf-decomposition}]
	The statement is trivial for $n=1$ and for $n>1$ we apply \eqref{eq:shuffle-product} to \eqref{eq:word-reginf-decomposition} such that the right-hand side becomes
	\begin{multline*}
		\sum_{0 \leq k < i \leq n}^{n}
		\left( \letter{\sigma_i} - \letter{-1} \right)
		\left[ 
			\letter{-1}^{k}
			\shuffle
			(-\letter{-1})^{i-k-1}
			\shuffle
			\letter{\sigma_{i+1}}\ldots\letter{\sigma_n}
		\right]
		\\
		+
		\letter{-1} 
		\sum_{k=1}^{n} \left[
			\letter{-1}^{k-1} 
			\shuffle 
			\WordReg{}{\infty}\left( \letter{\sigma_{k+1}}\ldots\letter{\sigma_n} \right)
		\right].
	\end{multline*}
	The second contribution is $\letter{-1} \letter{\sigma_2}\ldots \letter{\sigma_n}$ by the induction hypothesis, while the sum $\sum_{k=0}^{i-1} \letter{-1}^k \shuffle (-\letter{-1})^{i-1-k} = (\letter{-1} - \letter{-1})^{\shuffle(i-1)} = \delta_{i,1}$ reduces the first contribution to $\left( \letter{\sigma_1} - \letter{-1} \right) \letter{\sigma_2}\ldots\letter{\sigma_n}$.
\end{proof}

\begin{proof}[Lemma \ref{lemma:contour-splitting}]
	Consider a sequence $\letter{\tau}^n$ in a word $w=u (\letter{\tau}^n) v$ where $u \in \im(\WordReg{\tau}{})$ does not end in $\letter{\tau}$ and $v \in \im(\WordReg{}{\tau})$ does not begin with $\letter{\tau}$. The contributions to \eqref{eq:chens-lemma-contour-splitting} that split $w$ between $u$ and $v$ are
	\begin{equation}
		\sum_{k=0}^{n} 
			\iInt_{\gamma_u}  \left( u\letter{\tau}^k \right)
			\cdot
			\iInt_{\eta_u} \left( \letter{\tau}^{n-k}v \right).
			\tag{$*$}%
			\label{eq:contour-splitting-proof-deconcat} %
	\end{equation}
	From lemma~\ref{lemma:shuffle-head-tail-decomposition} we obtain the identities
	\begin{align}
		u (\letter{\tau}^k)
		&=
		\sum_{\mu=0}^{k} 
			\letter{\tau}^{k-\mu} 
			\shuffle 
			\WordReg{\tau}{} \left( u\letter{\tau}^{\mu} \right)
		\quad\text{and}
		\label{eq:regtail-shuffle-identity} %
		\\
		(\letter{\tau}^k) v
		&=
		\sum_{\nu=0}^{k}
			\letter{\tau}^{k-\nu}
			\shuffle
			\WordReg{}{\tau} \left( \letter{\tau}^{\nu}v \right)
		,
		\label{eq:reghead-shuffle-identity} %
	\end{align}
	which allow us to rewrite \eqref{eq:contour-splitting-proof-deconcat} as
	\begin{equation*}
		\sum_{\mu+\nu+a+b=n} 
				\iInt_{\gamma_u} \WordReg{\tau}{}\left( u\letter{\tau}^{\mu} \right)
				\cdot
				\iInt_{\eta_u} \WordReg{}{\tau}(\letter{\tau}^{\nu} v)
				\cdot
				\iInt_{\gamma_u} \letter{\tau}^{a}
				\cdot
				\iInt_{\eta_u} \letter{\tau}^{b}
		.
	\end{equation*}
	The sum over $a+b=n-\mu-\nu$ of the last two terms combines to $\iInt_{\gamma} \letter{\tau}^{n-\mu-\nu}$. Now the limit $u\rightarrow \tau$ in the remaining two factors is finite, such that \eqref{eq:contour-splitting-proof-deconcat} then becomes
	\begin{multline*}
		\sum_{\mu+\nu \leq n}
		\iInt_{\tau}^{z} \WordReg{\tau}{} \left( u\letter{\tau}^{\mu} \right)
		\cdot
		\iInt_{\tau}^{z} \letter{\tau}^{n-\mu-\nu}
		\cdot
		\iInt_0^{\tau} \WordReg{}{\tau} \left( \letter{\tau}^{\nu} v \right)
		\\
		=
		\sum_{\mu=0}^{n}
		\iInt_{\tau}^{z} \left( u \letter{\tau}^{\mu} \right)
		\cdot
		\iInt_{0}^{\tau} \WordReg{}{\tau}\left( \letter{\tau}^{n-\mu} v \right)
		.
	\end{multline*}
	We used \eqref{eq:regtail-shuffle-identity} again and the definition $\iInt_{\tau}^z \letter{\tau} \defas \iInt_{\gamma} \letter{\tau}$.
\end{proof}

\section{List of functions and options provided by \HyperProg}
\label{sec:reference}%
\subsection{Options and global variables}
\begin{description}
	\item[\code{\_hyper\_verbosity}]
		(default: $1$) \\
		The higher this integer, the more progress information is printed during calculations. The value zero means no such output at all.

	\item[\code{\_hyper\_verbose\_frequency}]
		(default: $10$) \\
		Sets how often progress output is produced during integration or polynomial reduction.

	\item[\code{\_hyper\_return\_tables}]
		(default: \code{false}) \\
		When \code{true}, \code{integrationStep} returns a table instead of a list. This is useful for huge calculations, because {\Maple} can not work with long lists.
	\item[\code{\_hyper\_check\_divergences}]
		(default: \code{true}) \\
		When active, endpoint singularities at $z\rightarrow 0,\infty$ are detected in the computation of integrals $\int_0^{\infty} f(z)\ \dd z$.

	\item[\code{\_hyper\_max\_pole\_order}]
		(default: $10$) \\
		Sets the maximum values of $i$ and $j$ in \eqref{eq:check-divergences} for which the functions $f_{i,j}$ are computed to check for potential divergences $f_{i,j} \neq 0$.

	\item[\code{\_hyper\_abort\_on\_divergence}]
		(default: \code{true}) \\
		This option is useful when divergences are detected erroneously, as happens when periods occur for which no basis is supplied to the program.

	\item[\code{\_hyper\_divergences}]
		\hfill \\
		A table collecting all divergences that occurred.

	\item[\code{\_hyper\_splitting\_field}]
		(default: $\emptyset$) \\
		This set $R$ of radicals defines the field $k=\Q(R)$ over which all factorizations are performed.

	\item[\code{\_hyper\_ignore\_nonlinear\_polynomials}]
		(default: \code{false}) \\
		Set to \code{true}, all non-linear polynomials (that would result in algebraic zeros as letters) will be dropped during integration. This is permissible when linear reducibility is granted.

	\item[\code{\_hyper\_restrict\_singularities}]
		(default: \code{false}) \\
		When \code{true}, the rewriting of $f$ as a hyperlogarithm in $z$ (performed during integration) projects onto the algebra $L(\Sigma)$ of letters $\Sigma$ specified by the roots of the set \code{\_hyper\_allowed\_singularities} (default: $\emptyset$) of irreducible polynomials. This can speed up the integration.

	\item[\code{\_hyper\_algebraic\_roots}]
		(default: \code{false}) \\
		When true, all polynomials will be factored linearly which can introduces algebraic functions. Further computations with such functions are not supported. 
\end{description}

\subsection{{\Maple} functions extended by {\HyperProg}}
\begin{description}
	\item[{\code{convert}$(f, \code{form})$}] with $\code{form} \in \set{\code{Hlog}, \code{Mpl}, \code{HlogRegInf}}$\\
		Rewrites polylogarithms $f$ in terms of hyper- or polylogarithms using \eqref{eq:Hyper-as-Li}. Choosing $\code{form}=\code{HlogRegInf}$ transforms $f$ into the list representation \eqref{eq:maple-internal-reginf-rep}.

	\item[{\code{diff}$\left( f, z \right)$}] \hfill\\
		Computes the partial derivative $\partial_t f$ of hyperlogarithms $\code{Hlog}\left(g(t), w(t)\right)$ or polylogarithms $\code{Mpl}\left( \vec{n},\vec{z}(t) \right)$ that occur in $f$. This works completely generally, i.e.\ also when a word $w(t)$ depends on $t$.

	\item[{\code{series}$\left(f, z=0 \right)$}] \hfill\\
		Implements the expansion of $f=\Hyper{w}(z)$ at $z\rightarrow 0$. To expand at different points, use $\code{fibrationBasis}$ first as explained in the manual.
\end{description}

\subsection{Some new functions provided by {\HyperProg}}
Note that there are further functions in the package, cf. the manual.
\begin{description}
	\item[{\code{integrationStep}$\left( f, z \right)$}] \hfill\\
		Computes $\int_0^{\infty} f\ \dd z$ for $f$ in the form \eqref{eq:maple-internal-reginf-rep}.

	\item[{\code{hyperInt}$\left(f, \vec{z}\right)$}] with a list $\vec{z}=[z_1,\ldots,z_r]$ or single $\vec{z} = z_1$\\
		Computes $\int_0^{\infty} \dd z_r \ldots \int_0^{\infty} \dd z_1 f$ from right to left. Any variable can also specify the bounds $z_i=a_i..b_i$ to compute $\int_{a_i}^{b_i} \dd z_i$ instead.

	\item[{\code{fibrationBasis}$\left(f, [z_1,\ldots,z_r], F, S\right)$}] \hfill\\
		Rewrites $f$ as an element of $L(\Sigma_1)(z_1) \tp \ldots \tp L(\Sigma_r)(z_r)$. Note that $\Sigma_i \subset \overline{\C(z_{i+1},\ldots,z_r)}$ in general are algebraic functions of the following variables. 
		A table $F$ (with indexing function \code{sparsereduced}) can be supplied to store the result in compact form, otherwise \code{Hlog}-expressions are returned.
		
		For each defined key $z_i$ of $S$, the result is projected from $L(\Sigma_i)(z_i)$ onto $L(\Sigma_i^S)(z_i)$ restricting to letters $\Sigma_i^S \defas \setexp{\text{zeros of $p(z_i)$}}{p \in S_{z_i}}$. All words including other letters are dropped in the computation.

	\item[{\code{index/sparsereduced}}] \hfill\\
		This indexing function corresponds to {\Maple}s \code{sparse}, but entries with value zero are removed. It is used to collect coefficients of hyperlogarithms.

	\item[{\code{forgetAll}$()$}] \hfill\\
		Clears cache tables for internal functions and should be called when options were changed.

	\item[{\code{transformWord}$(w, t)$}] \hfill\\
		Given a word $w=[\sigma_1,\ldots,\sigma_n]$ as a list, this function returns a list $[ [w_1, u_1], \ldots]$ of pairs such that
		\begin{equation*}
			\AnaReg{z}{\infty} \Hyper{w}(z)
			=
			\sum_i \Hyper{w_i}(t) \cdot \AnaReg{z}{\infty} \Hyper{u_i}(z)
		\end{equation*}
		and implements the algorithm of section~\ref{sec:reginf-as-hyperlog}. Note that $u_i$ is given in the product form \eqref{eq:maple-internal-reginf-rep}.

	\item[{\code{reglimWord}$(w, t)$}] \hfill \\
		Given a word $w=[\sigma_1,\ldots,\sigma_n] \in \Sigma^{\times}$ with rational letters  $\Sigma \subset \C(t)$, it implements our algorithm from section~\ref{sec:reglim-algorithm} to compute $u$ (in the representation \eqref{eq:maple-internal-reginf-rep}) such that
		\begin{equation*}
			\AnaReg{t}{0} \AnaReg{z}{\infty} \Hyper{w}(z)
			=
			\AnaReg{z}{\infty} \Hyper{u}(z).
		\end{equation*}

	\item[{\code{integrate}$(f, z)$}] \hfill \\
		Takes a hyperlogarithm $f(z)$ in the form \eqref{eq:maple-internal-hyperlog} and computes a primitive $F$ ($\partial_z F(z) = f(z)$) following section~\ref{sec:integration-differentiation}.

	\item[{\code{cgReduction}$\left( L, \code{todo}, d \right)$}] \hfill \\
		Computes compatibility graphs (stored in the table $L$) for the variable sets in the list \code{todo}, only taking reductions into account where every polynomial is of degree $d$ or less (default $d=1$).
		
		When \code{todo} is a set, then all reductions of variables not overlapping this set are computed.

	\item[{\code{checkIntegrationOrder}$\left( L, \vec{z} \right)$}] \hfill\\
		Tests whether for $\vec{z}=[z_1,\ldots]$, all polynomials in the reduction $L$ are linear in the corresponding $z_i$ and prints the number of polynomials.
\end{description}

\subsection{Functions related to Feynman integrals}
\label{sec:function-list-feynman}%
\begin{description}
	\item[{\code{graphPolynomial}$(E)$}] \hfill\\
		Computes $\psipol$ for the graph with edges $E=[ e_1,\ldots ]$ given as a list of pairs of vertices $e_i=[v_{i,1}, v_{i,2}]$. The vertices $V=\set{1,\ldots,\abs{V}}$ must be numbered consecutively.

	\item[{\code{forestPolynomial}$(E, P)$}] \hfill\\
		The spanning forest polynomial $\forestpolynom{P}$ of \cite{BrownYeats:SpanningForestPolynomials} of the graph $E$, $P$ is a partition of a subset of vertices.

	\item[{\code{secondPolynomial}$(E, P, M)$}] \hfill\\
		Computes $\phipol$ for the graph with edges $E$ that denote scalar propagators of masses $M$ (optional). $P=[ [v_1,p_1^2],\ldots]$ lists the vertices $v_i$ to which external momenta $p_i$ are attached.

	\item[{\code{graphicalFunction}$(E,V_{\text{ext}})$}] \hfill\\
Constructs the parametric integrand for a \emph{graphical function} as defined in \cite{Schnetz:GraphicalFunctions} in $\Dim=4$. The edge list $E=[e_1,\ldots]$ can contain lists $e_i=[v_{i,1},v_{i,2}]$ to denote propagators and sets $e_i=\set{v_{i,1},v_{i,2}}$ to denote inverse (numerator) propagators.
	The external vertices are $V_{\text{ext}} = [v_z,v_0,v_1,v_{\infty}]$ with $v_{\infty}$ optional.

	\item[{\code{drawGraph}$(E, P, M, s)$}] \hfill\\
Draws the graph defined by the edge list $E$. The remaining parameters are optional: $P$ and $M$ are as for \code{secondPolynomial} while $s$ sets the style of the drawing (see \code{GraphTheory[DrawGraph]}).

	\item[{\code{findDivergences}$\left( f, P \right)$}] \hfill\\
		For any pair $J \cap K=\emptyset$ of disjoint sets of variables, the degree $\omega_J^K(f)$ of divergence when $z\rightarrow 0,\infty$ (for $z \in J,K$) is computed as defined in \cite{Panzer:DivergencesManyScales}. The result is a table indexed by the sets $J \cupdot K^{-1}$, holding the values of $\omega_J^K(f)$ that are $\leq 0$ when $\varepsilon=0$.

		The variables $P$ are considered fixed parameters, so only sets with $(J\cupdot K) \cap P = \emptyset$ are considered.

	\item[{\code{dimregPartial}$\left( f, I, \text{sdd} \right)$}] \hfill\\
		Computes the new integrand $\mathcal{D}_J^K(f)$ after a partial integration, as defined in \cite{Panzer:DivergencesManyScales}:
		\begin{equation}
			\mathcal{D}_J^K
			= 1 - \frac{1}{\text{sdd}} \left[ 
				\sum_{z \in J} \partial_z z
				-
				\sum_{z \in K} \partial_z z
			\right].
			\label{eq:partial-integration-operator}%
		\end{equation}
		The set $I=J\cupdot K^{-1}$ must consist of variables $J$ and inverses $K^{-1} = \setexp{z^{-1}}{z\in K}$.
\end{description}

\phantomsection
\pdfbookmark[1]{References}{final-bibliography}
\bibliographystyle{elsarticle-num}
\bibliography{../qft}

\end{document}